\documentclass[12pt]{article}
\usepackage[utf8]{inputenc}
\usepackage{amsfonts}
\usepackage{amssymb}
\usepackage{amsmath}
\usepackage{amstext}
\usepackage{amsthm}
\usepackage{chngpage}
\usepackage{graphicx}
\usepackage[margin=1.25in]{geometry}
\usepackage[colorlinks=false,pdfborder={0 0 0}]{hyperref}
\usepackage[longnamesfirst]{natbib}
\usepackage{setspace}
\usepackage{multirow}
\usepackage{booktabs}
\usepackage{threeparttable}
\usepackage{xcolor}
\usepackage{lscape}
\providecommand{\U}[1]{\protect\rule{.1in}{.1in}}
\theoremstyle{plain}
\newtheorem{theorem}{Theorem}
\newtheorem{assumption}{Assumption}
\theoremstyle{definition}

\newtheorem{inneruprocedure}{Procedure}
\newenvironment{uprocedure}[1]{\renewcommand\theinneruprocedure{#1}\inneruprocedure}{\endinneruprocedure}

\newcommand{\bge}{\begin{equation}}
\newcommand{\ene}{\end{equation}}

\def\monthname{\ifcase\month\or January\or February\or March\or April\or May\or June\or July\or August\or September\or October\or November\or December\fi}

\newcommand{\1}{\mbox{$\mathrm{1\hspace*{-2.5pt}l}$\,}}

\begin{document}

\thispagestyle{empty}

\begin{center}
{\Large Doubly Robust Estimation of Local Average Treatment Effects Using Inverse Probability Weighted Regression Adjustment\footnote{This article supersedes ``Doubly Robust IV Estimation of Local Average Treatment Effects'' by S. Derya Uysal.}}

\bigskip\bigskip\bigskip\bigskip

Tymon S\l oczy\'{n}ski

Brandeis University

\bigskip\bigskip

S.\ Derya Uysal

Ludwig Maximilian University of Munich

\bigskip\bigskip

Jeffrey M. Wooldridge

Michigan State University

\bigskip\bigskip\bigskip

This version: November 14, 2022
\end{center}

\begin{verbatim}
 
\end{verbatim}

\noindent
\textbf{Abstract:} We revisit the problem of estimating the local average
treatment effect (LATE) and the local average treatment effect on the
treated (LATT) when control variables are available, either to render the
instrumental variable (IV) suitably exogenous or to improve precision.
Unlike previous approaches, our doubly robust (DR) estimation procedures use
quasi-likelihood methods weighted by the inverse of the IV propensity score -- so-called inverse probability weighted regression adjustment (IPWRA)
estimators. By properly choosing models for the propensity score and outcome
models, fitted values are ensured to be in the logical range determined by
the response variable, producing DR estimators of LATE and LATT with
appealing small sample properties. Inference is relatively straightforward
both analytically and using the nonparametric bootstrap. Our DR LATE and DR
LATT estimators work well in simulations. We also propose a DR version of
the Hausman test that can be used to assess the unconfoundedness assumption through a comparison of different estimates of the average treatment effect on the treated (ATT) under
one-sided noncompliance. Unlike the usual test that compares OLS and IV estimates, this procedure is robust to treatment effect heterogeneity. \\
\phantom{} \\
\phantom{} \\
\phantom{} \\
\phantom{} \\
\phantom{} \\
\phantom{} \\
\pagebreak

\defcitealias{IA1994}{IA (1994)}
\defcitealias{DHL2014b}{DHL (2014)}
\shortcites{BBMT2022, BCFVH2017, Chernozhukovetal2018, Finkelsteinetal2012, TAWBF2014}

\doublespacing
\setcounter{page}{2}

\section{Introduction}

\label{intro}

Instrumental variables estimation of causal effects has a long history in
applied econometrics. In introductory econometrics courses, the properties
of the instrumental variables (IV)\ estimator are often taught within the
framework of a linear model with a constant coefficient. When applied in a
treatment effects setting, the constant coefficient assumption is equivalent
to assuming a constant treatment effect in the population. In their
pioneering work, \cite{IA1994} [\citetalias{IA1994}] used a potential
outcomes framework to study the probability limit of the simple IV estimator
in the setting of a binary treatment and binary instrumental variable. \citetalias{IA1994} showed that, under reasonable assumptions, the IV estimator
consistently estimates a parameter now known as the \textit{local average
treatment effect (LATE)}, which is the average treatment effect over the
subpopulation of units that comply with their randomized eligibility.
\cite{AIR1996} explicitly embedded the LATE setup within the
setting of the Rubin Causal Model, showed how the IV estimand identifies a
causal parameter under certain assumptions, and discussed the consequences of
violations of those assumptions. \cite{Vytlacil2002} demonstrated that the LATE framework is equivalent to a nonparametric selection model with a weakly separable selection equation.

In many applications of IV, the instrument is not randomly assigned, in
which case the simple IV estimator -- also known as the \textit{Wald
estimator} -- is no longer consistent for LATE\@. In some cases, conditioning
on observed covariates, or controls, can render the IV as good as randomly
assigned within subpopulations defined by the covariate values. In effect,
the IV is assumed to satisfy an unconfoundedness assumption conditional on
observables. In most textbook treatments of instrumental variables that
include control variables $X$, these are added linearly and
then they act as their own instruments. In the treatment effects context,
it may seem appealing to include interactions of the binary treatment variable, $W$, with
the (suitably centered) covariates $X$\@. If $Z$ is the
instrument for $W$, a natural instrument for $W \cdot X$ is $Z \cdot X$\@. \citet[Procedure 21.2]{Wooldridge2010} describes an IV procedure that exploits the
binary nature of $W$ by using a binary response model for $p\left(
Z,X\right) \equiv \mathbb{P}\left( W=1|Z,X\right) $ and then using $\hat{p}%
\left( Z_{i},X_{i}\right) $ and $\hat{p}\left( Z_{i},X_{i}\right) \cdot
X_{i}$ as instruments in the linear equation that includes $W_{i}$ and $%
W_{i}\cdot X_{i}$.

Adding covariates to a linear equation and interacting them with the
treatment indicator seems like a natural way to account for nonrandom
assignment of the instrument while allowing for heterogeneous treatment effects. Unfortunately, no results imply that this procedure generally
uncovers the LATE\@. By contrast, \cite{Tan2006} and \cite{Frolich2007} independently
obtained a useful identification result for LATE when covariates are needed
in order to render the IVs ignorable. \cite{Frolich2007} used his identification
result to obtain consistent, asymptotically normal estimators of LATE\@. As a
practical matter, however, the need to estimate four conditional mean functions
nonparametrically makes \citeauthor{Frolich2007}'s estimator difficult to implement, even
with just a small number of covariates. Plus, issues of how to handle
discrete, continuous, and mixed control variables need to be addressed.

On the other hand, the
estimation approach proposed by \cite{Tan2006} is based on so-called \textit{%
augmented inverse probability weighting (AIPW)} estimators. AIPW is a standard class of doubly robust (DR) estimators, that is, estimators that remain consistent under misspecification of either of the two (sets of) parametric working models on which they are based. However, as discussed in \cite{KS2007}, AIPW estimators, as commonly applied, are often unstable in practice, as is standard \textit{inverse probability
weighting (IPW)}\@. One reason is that these estimators are often based on weights in the weighted averages that do not sum to unity; in other
words, the weights are not normalized.

Following \cite{Tan2006} and \cite{Frolich2007}, many other estimation approaches for LATE have
been proposed, some of which are doubly robust and some are not. For
example, \cite{DHL2014b} [\citetalias{DHL2014b}] studied the aforementioned IPW estimators, which are consistent when the instrument
propensity score is correctly specified but not otherwise. Admittedly, \citetalias{DHL2014b} suggested
nonparametric series estimators, which in theory resolves the issues of misspecification, but in practice their approach would be applied in a flexible parametric framework by most practitioners. Similar to \cite{Tan2006}, other DR estimators have also been based on the AIPW approach, and this includes both those proposed in \cite{ORR2015} and several estimators that employ high-dimensional selection, including those in \cite{BCFVH2017}, \cite{Chernozhukovetal2018}, and \cite{ST2022}, which often also rely on sample splitting to allow for high-dimensional covariates. In recent work, \cite{Heiler2022} discussed a DR extension of a particular balancing estimator of LATE while \cite{SS2022} combined ``kappa weighting'' \citep{Abadie2003} and high-dimensional selection to obtain DR estimators of LATE and related parameters.

Our primary purpose in this paper is to propose a new class of doubly robust estimators of LATE that are simple to implement and avoid the
shortcomings of nonparametric conditional mean estimation and AIPW methods.
In particular, using the identification result in \cite{Frolich2007} and building also on \cite{Wooldridge2007} and \cite{SW2018}, we
show how estimators that use the inverse of the instrument propensity score
to weight the objective functions for estimating the treatment propensity
score and the conditional mean of the response allow consistent estimation
of LATE\@. These estimators, now commonly labeled \textit{inverse probability
weighted regression adjustment (IPWRA)} estimators, have the same double
robustness property of AIPW estimators. An advantage of IPWRA estimators is
that one can choose functional forms so that the estimated conditional
probability and conditional mean functions are guaranteed to produce
predictions within the logical range of the outcomes. This feature of IPWRA
makes the resulting estimators of LATE have good finite sample properties.
Moreover, the estimators are easy to obtain and inference is relatively
straightforward. Using a similar approach, we also propose DR estimators for the
local average treatment effect on the treated (LATT)\@.

In Section \ref{identification} we provide the setting, define the LATE
parameter, and summarize the identification result in \cite{Frolich2007}. We also study identification of LATT, beginning with a result
due to \cite{FL2010} but modifying it to obtain a simple
representation of this parameter that leads naturally to DR estimation.

Section \ref{robust} shows how the IPWRA approach can be used to identify
the four expectations appearing in LATE when conditioning on covariates. We
modify existing arguments to account for the fact that the conditional means
we need to estimate for the outcome are not of the potential outcomes.
Nevertheless, the IPWRA approach still identifies the required unconditional
means. This section carries out a similar analysis for LATT where we are able to relax the
assumptions used to identify LATE\@.

Section \ref{inference} shows how to obtain standard errors for the DR LATE
and DR LATT estimators that account for the sampling error in all estimation
steps. Section \ref{testing} shows how to modify the Hausman-type test
proposed by \citetalias{DHL2014b} to allow for DR estimation. In the case with
one-sided noncompliance, if assignment is unconfounded then LATT is the same as the average
treatment effect on the treated (ATT)\@. Therefore, we can obtain two DR
estimators using IPWRA estimation schemes: one that uses an instrumental
variable and another that employs unconfoundedness conditional on
covariates. We show how to test the null hypothesis that the two estimators
consistently estimate the same parameter.

Section \ref{applications} revisits two empirical studies. First, we use the data in \cite{Abadie2003} to
produce LATE and LATT estimates for the effect of
participating in a 401(k) pension plan on net financial wealth. We also look at the causal effect of
participating in a 401(k) plan on participation in individual retirement
accounts (IRAs). In this case, we compare IV estimates of a linear
probability model with our DR estimates that recognize the binary nature of
the IRA\ participation decision. Using our proposed method, we find that 401(k) participation has a positive effect on net financial assets and the probability of IRA participation. This is despite the fact that the corresponding AIPW (for net financial assets) and IPW estimates (for both outcomes) are much smaller and imprecise. Even though both AIPW and IPWRA are doubly robust, they lead to different conclusions about the LATE on net financial assets.

Second, we also revisit \cite{Finkelsteinetal2012} and \cite{TAWBF2014}, and use the data from the Oregon Health Insurance Experiment to study the effects of Medicaid on emergency room visits. Like in previous work, our estimates are positive, which suggests, perhaps counterintuitively, that access to health insurance may increase the utilization of emergency rooms. Our novel empirical contribution is that LATT, the effect on the treated compliers, is larger than the usual LATE, at least along the extensive margin. This is because treatment effects appear to be more pronounced in larger households \citep[cf.][]{DL2022}, which are also more likely to be treated.

Section \ref{simulations} presents simulation evidence on the performance of several estimators of LATE, including IV, regression adjustment (RA), IPW, AIPW, and IPWRA\@. The performance of our proposed method, IPWRA, is very satisfying. It is never substantially more biased than the competing estimators while its precision is better than that of AIPW, which is the only alternative that shares the double robustness property of IPWRA\@. Finally, Section \ref{conclusion} concludes.

\section{Identification of LATE and LATT}

\label{identification}

The potential outcomes setting in this paper is the one pioneered by \citetalias{IA1994}. Eventually, we will assume access to a random sample from the
population, and so all assumptions can be stated in terms of random
variables representing the population of interest.

For a binary intervention, let $Y(0)$ be the potential outcome
in the control state and $Y(1)$ the potential outcome in the
treated state. The observed binary treatment indicator is $W$, where $W=1$
denotes treatment and $W=0$ denotes control. We have access to a binary
instrumental variable, $Z$. As in \citetalias{IA1994}, there are potential treatment
statuses based on the assignment of the instrument (which is often
eligibility): $W(1)$ is participation status when a unit is
made \textquotedblleft eligible\textquotedblright\ and $W(0)$
is participation status in the \textquotedblleft
ineligible\textquotedblright\ state. This framework allows for the
possibility that units do not comply with their assigned \textquotedblleft
eligibility.\textquotedblright\ For example, some workers, if selected to
participate in a job training program ($Z=1$), may choose not to participate
[$W(1)=0$].

The observed outcome $Y$ is a function of the observed treatment variable
and the potential outcomes corresponding to the treatment and control
status: 
\begin{equation*}
Y=WY(1) +(1-W)Y(0) .
\end{equation*}%
Further, the realized treatment status can be written in terms of the
instrument $Z$ and the potential treatment statuses: 
\begin{equation}
W=ZW(1) +(1-Z)W(0) .  \label{pottreat}
\end{equation}

\noindent
According to the relationship between the potential treatment status and the
binary instrument, the population can be divided into four subpopulations:
compliers, always-takers, never-takers, and defiers. From the observed dataset one cannot identify the group to which an individual belongs since only
the pair $(W,Z)$ is observed. For example, if $Z=1$ and $W=1$, the
individual is either a complier or an always-taker. Always-takers and never-takers do not change their treatment behavior when the assignment of the IV
changes. The only subpopulations that can be induced into changing $W$
through a variation in $Z$ are the defiers and compliers.

Generally, the treatment effect of interest can be defined either as the
impact of the treatment on the outcome for the defiers [$W(1)
<W(0) $] or for the compliers [$W(1) >W\left(
0\right) $]. Following the literature, we focus on average treatment effects
on compliers under the assumption that defiers do not exist. The local
average treatment effect (LATE) is the expected difference between the
potential outcomes for the subpopulation of compliers:%
\begin{equation}
\tau _{LATE}=\mathbb{E}[Y(1) -Y(0) |W(1)
>W(0) ].  \label{LATE}
\end{equation}%
Compliers are members of a hypothetically defined subpopulation and cannot
be identified from observed data without further assumptions.

As in much of the literature since \citetalias{IA1994} -- including
several papers discussed in the introduction -- we assume that we have
(pre-treatment) covariates, $X$, that render the instrumental variable
suitably exogenous when conditioned on. The support of $X$ is indicated by $%
\mathcal{X}$. With these covariates, identification of $\tau _{LATE}$ is
possible if certain assumptions are met.

The first assumption is that, conditional on $X$, the instrumental variable
has no direct effect on the potential outcomes; its effect can come only
through the treatment assignment. The formal statement requires indicating
two arguments in the potential outcomes, $Y ( w,z )$ for $w,z\in
\{0,1\}$.

\begin{assumption}[Exclusion Restriction]
\label{ass1}  For $w\in \{0,1\}$
and almost all $x\in \mathcal{X}$,
\begin{equation}
	\mathbb{P}\left[Y ( w,1 ) = Y ( w,0 ) \mid X=x \right]=1.\text{ }\square
	\label{exclusion1}
\end{equation}
\end{assumption}

\noindent Assumption \ref{ass1} justifies labeling the potential outcomes
using a single index that indicates actual treatment status because we
condition on $X$ in stating ignorability of the instruments. In what
follows, we use $Y ( w )$ as the potential outcome for treatment
status $w\in \left\{ 0,1\right\} $.

\begin{assumption}[Ignorability of Instrument]\label{ass2}
Conditional
on $X$, the potential outcomes are jointly independent of $Z$:
\begin{equation}
	\left[ Y(0) ,Y(1) ,W(0) ,W(1) \right] \perp Z \mid X.\text{ }\square  \label{ignor1}
\end{equation}%
\end{assumption}

\noindent Assumption \ref{ass2} requires that, conditional on observed confounders,
the instrument can be regarded as random.

\begin{assumption}[Monotonicity]\label{ass3}
\begin{equation*}
	\mathbb{P}[W(1) \geq W(0) ]=1.\text{ }\square
\end{equation*}%
\end{assumption}

\noindent This monotonicity assumption is standard in the literature: it says that
there are no defiers in the population (or that the group is so small it has
probability zero). It is equivalent to a conditional statement, namely, $%
\mathbb{P}[W(1) \geq W(0) |X=x]=1$ for almost all $%
x\in \mathcal{X}$. In other words, if Assumption \ref{ass3} holds, $\mathbb{P%
}[W(1) \geq W(0) |X=x]<1$ is possible only on a
subset of $x\in \mathcal{X}$ with measure zero. Formally, this claim is
equivalent to the proposition that for a random variable $R\geq 0$, $\mathbb{%
E}\left( R\right) =0$ if and only if $\mathbb{P}\left( R=0\right) =1$.

The next assumption requires the existence of compliers in the population.

\begin{assumption}[Existence of Compliers]\label{ass4}
\begin{equation*}
	\mathbb{P}[W(1) >W(0) ]>0.\text{ }\square
\end{equation*}%
\end{assumption}

\noindent
When we partition the population on the basis of the covariates $X$,
Assumption \ref{ass4} implies that, for some subset $\mathcal{C\subset X}$
with $\mathbb{P}\left( \mathcal{C}\right) >0$, $\mathbb{P}[W(1)
>W(0) |X=x]>0$ if $x\in \mathcal{C}$. This follows by iterated
expectations: If $\1\left[ W(1) >W(0) \right] $ has
positive expectation then its expectation conditional on $X$ must be
positive with nonzero probability. The subset $\mathcal{C}$ defines the
subpopulation of compliers based on the values of $X$.

The requisite overlap assumption is stated in terms of the propensity score
involving the instrumental variable, sometimes referred to as the \textit{%
instrument propensity score}.

\begin{assumption}[Overlap for LATE]\label{ass5}
For almost all $x\in 
\mathcal{X}$,
\begin{equation*}
	0<\mathbb{P}\left( Z=1|X=x\right) <1.\text{ }\square
\end{equation*}
\end{assumption}

\noindent
\citetalias{IA1994} show that if Assumptions \ref{ass1}--\ref{ass5}
hold without conditioning on $X$, then $\tau _{LATE}$ is identified as
\begin{equation}
\tau _{LATE}=\frac{\mathbb{E}[Y|Z=1]-\mathbb{E}[Y|Z=0]}{\mathbb{E}[W|Z=1]-%
\mathbb{E}[W|Z=0]}.  \label{lateid}
\end{equation}%
In that case, given a random sample from the population, $\tau _{LATE}$ can
be consistently estimated by replacing the expectations in (\ref{lateid}) with
the corresponding sample averages. This simple estimator is the well-known Wald estimator;
it is also obtained by estimating the simple linear equation $Y=\alpha
+\beta W+U$ by instrumental variables using instruments $\left( 1,Z\right) $%
. In other words, $\hat{\beta}_{IV}=\hat{\tau}_{LATE}$. In some cases, one
cannot participate ($W=1$) unless assigned to the treatment ($Z=1$), in
which case the second term in the denominator of $\tau _{LATE}$ is zero: $%
\mathbb{E}[W|Z=0]=\mathbb{P}[W=1|Z=0]=0$. This is the case in the
application in \cite{Abadie2003}  (also used in several subsequent studies), where
an employee cannot participate in an employer-sponsored pension plan unless
the employer offers such a plan. In other applications, $Z=1$ implies $W=1$
(no never-takers), in which case the first term in the denominator of $\tau
_{LATE}$ is unity. This situation arises in \cite{AE1998}  when,
in a population of women with at least two children, the \textquotedblleft
treatment\textquotedblright\ is having more than two children and the binary
instrument indicates whether the second birth was a multiple birth.

In many cases the instrumental variable candidate, $Z$, is not truly
randomized, but we might be willing to assume it is as good as randomized
conditional on $X$\@. Then, Assumptions \ref{ass1}--\ref{ass5} imply that we
can identify $\tau _{LATE}$. The following theorem is due to \citet[Theorem 1]{Frolich2007}. \cite{FL2010} relax the assumptions
somewhat but not in a way that makes the theorem clearly more applicable.

\begin{theorem}[Identification of LATE]\label{theorem1}
 Under
Assumptions \ref{ass1}--\ref{ass5},
\begin{equation}
	\tau _{LATE}=\frac{\mathbb{E}\left[ {\mathbb{E}}\left( {Y|X,Z=1}\right) -{%
			\mathbb{E}}\left( {Y|X,Z=0}\right) \right] }{\mathbb{E}\left[ {\mathbb{E}}%
		\left( {W|X,Z=1}\right) -{\mathbb{E}}\left( {W|X,Z=0}\right) \right] }=\frac{%
		\mathbb{E}[\mu _{1}(X)-\mu _{0}(X)]}{\mathbb{E}[\rho _{1}(X)-\rho _{0}(X)]},
	\label{latecondid}
\end{equation}%
where 
\begin{eqnarray}
	\mu _{0}(X) &\equiv &\mathbb{E}\left( Y|X,Z=0\right)  \label{mu_0} \\
	\mu _{1}(X) &\equiv &\mathbb{E}\left( Y|X,Z=1\right)  \label{mu_1}
\end{eqnarray}%
and 
\begin{eqnarray}
	\mathbb{E}\left( W|X,Z=0\right) &=&\mathbb{E}[W(0) |X]=\rho
	_{0}(X)  \label{rho_0} \\
	\mathbb{E}\left( W|X,Z=1\right) &=&\mathbb{E}[W(1) |X]=\rho
	_{1}(X).\text{ }\square  \label{rho_1}
\end{eqnarray}
\end{theorem}

\noindent
As discussed in \cite{Frolich2007}, the result in equation (\ref{latecondid})
suggests that one can estimate each of the four conditional mean functions, $%
\mathbb{E}\left( Y|X,Z=0\right) $, $\mathbb{E}\left( Y|X,Z=1\right) $, $%
\mathbb{E}\left( W|X,Z=0\right) $, and $\mathbb{E}\left( W|X,Z=1\right) $,
using nonparametric methods, and then average the estimates across $i$ to
estimate the unconditional means. Especially when the dimension of $X$ is
large, nonparametric estimation is not attractive, and inference is also
complicated. One of the estimators considered by \cite{Tan2006} estimates the
numerator and denominator in (\ref{latecondid}) using augmented inverse
probability weighting (AIPW) estimators. AIPW estimators are popular in the
case of unconfounded assignment but simulations show they are not always
well behaved even in somewhat large samples \citep[e.g.,][]{KS2007}. \cite{SW2018}  provide a recent overview of the debate on the merits of AIPW
approaches. One issue that apparently has not been noted is that commonly
used AIPW estimators -- like those in \cite{Tan2006} -- implicitly use weights
in the weighted averages that do not sum to unity; in other words, the
weights are not normalized. Because many of the estimators summarized in the
introduction have an AIPW flavor, they suffer from the same problem --
whether they are based on parametric approaches or machine learning
algorithms. In the next section we show how to use a class of DR estimators based on weighted quasi-maximum likelihood estimation (QMLE) to estimate
the four means that appear in (\ref{latecondid}). This possibility was already indicated by \cite{SW2018} but without any of the details that we consider here.

There is also some interest in estimating the local average treatment effect
on the treated (LATT), formally defined as
\begin{equation}
\tau _{LATT}=\mathbb{E}[Y(1) -Y(0) \mid W(1)
>W(0) ,W=1].  \label{LATT}
\end{equation}%
\cite{FL2010} study identification of this parameter and
show that, under the assumptions in Theorem 1,
\begin{equation}
\tau _{LATT}=\frac{\mathbb{E}\left\{ [\mu _{1}(X)-\mu _{0}(X)] \eta (X) \right\} }{\mathbb{E}\left\{ [\rho _{1}(X)-\rho _{0}(X)] \eta (X) \right\} },  \label{lattcondid}
\end{equation}%
where $\eta (x)$ is the instrument propensity score:
\begin{equation}
\eta (x) \equiv \mathbb{P}\left( Z=1|X=x\right) \text{, }x\in 
\mathcal{X}.  \label{iv_ps}
\end{equation}%
For our purposes, we use a different representation of $\tau _{LATT}$. We
state this theorem under the exclusion and ignorability  assumptions in
Theorem 1 even though we could relax some of the assumptions. It is useful
to explicitly relax the overlap assumption.

\begin{assumption}[Overlap for LATT]\label{ass6}
For almost all $x \in \mathcal{X}$,
\begin{equation*}
	\mathbb{P}\left( Z=1|X=x\right) <1.\text{ }\square
\end{equation*}%
\end{assumption}

\noindent
The overlap assumption for LATT means that there can be subsets of the
population, based on the values of the control variables $X$, where units
are not eligible\ for the treatment.

\begin{theorem}[Identification of LATT]\label{theorem2}
Under
Assumptions \ref{ass1}--\ref{ass4} and Assumption \ref{ass6},
\begin{equation}
	\tau _{LATT}=\frac{\mathbb{E}\left( Y|Z=1\right) -\mathbb{E}[\mu _{0}(X)|Z=1]%
	}{\mathbb{E}\left( W|Z=1\right) -\mathbb{E}[\rho _{0}(X)|Z=1]}.\text{ }%
	\square  \label{lattcondid2}
\end{equation}
\end{theorem}

\begin{proof}

By iterated expectations and (\ref{lattcondid}),
\begin{equation*}
\tau _{LATT}=\frac{\mathbb{E}\left\{ [\mu _{1}(X)-\mu _{0}(X)]Z\right\} }{%
\mathbb{E}\left\{ [\rho _{1}(X)-\rho _{0}(X)]Z\right\} }
\end{equation*}%
and so, dividing the numerator and the denominator by $\mathbb{P}\left( Z=1\right)
>0 $,
\begin{equation*}
\tau _{LATT}=\frac{\mathbb{E}\left\{ [\mu _{1}(X)-\mu _{0}(X)]|Z=1\right\} }{%
\mathbb{E}\left\{ [\rho _{1}(X)-\rho _{0}(X)]|Z=1\right\} }.
\end{equation*}%
By definition of $\mu_{0} (x)$ and $\mu_{1} (x)$,
we can write
\begin{equation*}
Y=(1-Z)\mu _{0}(X)+Z\mu _{1}(X)+U,\text{ \ }\mathbb{E}\left( U|X,Z\right) =0.
\end{equation*}%
It follows that
\begin{equation*}
\mathbb{E}\left( Y|Z=1\right) =\mathbb{E}\left[ \mu _{1}(X)|Z=1\right].
\end{equation*}%
Also, $W=\left( 1-Z\right) W(0) +ZW(1) $ and so, by
ignorability,
\begin{equation*}
\mathbb{E}\left( W|X,Z=1\right) =\mathbb{E}\left[ W(1) |X,Z=1%
\right] =\mathbb{E}\left[ W(1) |X\right] =\rho _{1}(X).
\end{equation*}%
Iterated expectations implies $\mathbb{E}\left( W | Z=1\right) =\mathbb{E}%
\left[ \rho _{1}(X)|Z=1\right] $. Therefore, we can write $\tau _{LATT}$ as
in (\ref{lattcondid2}). The overlap assumption ensures that $\mathbb{E}[\mu
_{0}(X)|Z=1]$ is identified and ignorability and overlap ensure $\mathbb{E}%
[\rho _{0}(X)|Z=1]$ is identified. \end{proof}

\noindent
As we show in the next section, the representations in (\ref{latecondid}) and %
(\ref{lattcondid2}) permit doubly robust estimation of $\tau _{LATE}$ and $%
\tau _{LATT}$ using IPWRA estimators, providing a unification that makes the
estimation approaches transparent and simple.

\section{Doubly Robust Estimation of LATE and LATT}

\label{robust}

We now turn to estimation of $\tau _{LATE}$ and $\tau _{LATT}$ using a
particular class of doubly robust (DR) estimators, starting with the former.
The approach we take allows us to tailor the analysis based on the nature of
the observed outcome, $Y$, by choosing a suitable conditional mean function.
In particular, using the identification results in \cite{SW2018}, we extend \cite{Wooldridge2007}'s approach of combining inverse
probability weighting (IPW) and regression adjustment (RA) using a
particular quasi-maximum likelihood estimator. These estimators are commonly
referred to as IPWRA estimators.

\subsection{Estimation of LATE}

The identification result in equation (\ref{latecondid}) shows that, in
order to consistently estimate $\tau _{LATE}$, we need to consistently
estimate the following four quantities: 
\begin{eqnarray}
\theta _{1} &=&\mathbb{E}[\mu _{1}(X)]\text{, \ }\theta _{0}=\mathbb{E}[\mu
_{0}(X)],  \label{ucmean_mu} \\
\pi _{1} &=&\mathbb{E}[\rho _{1}(X)]\text{, \ }\pi _{0}=\mathbb{E}[\rho
_{0}(X)].  \label{ucmean_rho}
\end{eqnarray}%
Because of the representation of $W$ in equation (\ref{pottreat}), we can
immediately apply the DR results on IPWRA estimation from \cite{Wooldridge2007}
and \cite{SW2018}. The approach first requires
estimating a binary response model for the instrument propensity score
defined in (\ref{iv_ps}). By the overlap assumption (Assumption \ref{ass5}), $0<\eta
(x)<1 $ for all $x\in \mathcal{X}$. The proposal here is to use a standard
parametric model for $\eta (x)$, as is common in the literature when
estimating a propensity score function. Probably most popular is a flexible
logit model, but it could be a probit model, heteroskedastic probit model,
or something else. Let $G(x,\gamma )$ denote the parametric model for $\eta
(x)$. Under very general assumptions, the Bernoulli quasi-maximum likelihood
estimator, $\hat{\gamma}$, converges in probability to some value, $\gamma
^{\ast }$, which is sometimes called the quasi-true value or pseudo-true
value. If the model for $\eta (x)$ is correctly specified then $G(x,\gamma
^{\ast })=\mathbb{P}(Z=1|X=x)$. At this point, one would use the fitted
probabilities, $G(X_{i},\hat{\gamma})$, to study the LATE overlap condition
using standard methods; for a detailed discussion, see \citet[Chapter 14]{IR2015}.

After estimating the model for $\mathbb{P}(Z=1|X)$, next we estimate models
for $\rho _{0} ( x )$ and $\rho _{1} ( x )$, as defined
in (\ref{rho_0}) and (\ref{rho_1}). These are estimated by separate logit
models for $W$ for the $Z_{i}=0$ and $Z_{i}=1$ subgroups, applying the
inverse probability weights $1/\left[ 1-G(X_{i},\hat{\gamma})\right] 
$ and $1/G(X_{i},\hat{\gamma})$, respectively. The reason for using logit
models for $\rho _{0} ( x )$ and $\rho _{1} ( x )$ is to
ensure that the resulting estimators of the expected probabilities, $\pi
_{0} $ and $\pi _{1}$, are doubly robust, as discussed in \cite{Wooldridge2007}
and \cite{SW2018}. The logit function is the
canonical link function for the Bernoulli distribution, and that ensures the
DR\ property. Let $\Lambda ( \hat{\omega}_{0}+X_{i}\hat{\delta}%
_{0} )$ and $\Lambda ( \hat{\omega}_{1}+X_{i}\hat{\delta}%
_{1} )$ be the logit fitted values, where the estimated parameters are
obtained from the $Z_{i}=0$ and $Z_{i}=1$ subsamples, respectively. For
notational ease we show the indexes as linear functions of $X_{i}$ but,
naturally, any functions of the covariates may appear in the logit models.
In principle, one could use different functions of $X_{i}$, say $h_{0} (X_{i})$ and $h_{1} (X_{i})$, inside the logistic
function, but that seems to be rare in practice.

Having estimated the separate logit models by weighted Bernoulli QMLE, the
DR estimates of $\pi _{0}$ and $\pi _{1}$ are
\begin{equation*}
\hat{\pi}_{0}=N^{-1}\sum_{i=1}^{N}\Lambda ( \hat{\omega}_{0}+X_{i}\hat{%
\delta}_{0} ) \text{, \ }\hat{\pi}_{1}=N^{-1}\sum_{i=1}^{N}\Lambda
( \hat{\omega}_{1}+X_{i}\hat{\delta}_{1} ).
\end{equation*}%
From \cite{Wooldridge2007}, under standard regularity conditions, $\hat{\pi}_{z}$
is consistent for $\pi _{z}$ if the model for $\mathbb{P}(Z=1|X)$ is correct
or if the models for $\mathbb{P}(W=1|X,Z=0)$ and $\mathbb{P}(W=1|X,Z=1)$ are
correct. Naturally, if we know $\mathbb{P}(W=1|Z=1)=1$ then $\hat{\pi}_{1}$
is replaced with one and if we know $\mathbb{P}(W=1|Z=0)=0$ then $\hat{\pi}%
_{0}$ is replaced with zero (the more likely scenario when $Z$ is
eligibility and $W$ is participation).

Next, we show how to obtain DR estimators of $\theta _{0}$ and $\theta _{1}$%
. In doing so, it is useful to write $Y$ with a zero conditional mean error
term:
\begin{equation*}
Y=(1-Z)\mu _{0}(X)+Z\mu _{1}(X)+U,\text{ \ }\mathbb{E}\left( U|X,Z\right) =0,
\end{equation*}%
where $\mu _{0}(X)$ and $\mu _{1}(X)$ are defined in (\ref{mu_0}) and (\ref%
{mu_1}). Note that this is not the usual representation that leads to DR
estimation because $\mu _{0}(X)$ and $\mu _{1}(X)$ are not the potential
outcome conditional means; rather, these are the conditional mean functions
for the (observed) $Z=0$ and $Z=1$ subpopulations, respectively. Therefore,
we must modify the usual double robustness argument.

Let $m(\alpha _{0}+X\beta _{0})$ and $m(\alpha _{1}+X\beta _{1})$ be the
parametric models for $\mu _{0}(X)$ and $\mu _{1}(X)$. Again, for notational
ease we show these depending on an index linear in $X$, whereas they could
depend on (different) transformations of $X$ inside the function $m\left(
\cdot \right) $. We assume that the function $m\left( \cdot \right) $ is
based on the canonical link function for the chosen quasi-log likelihood (QLL) in
the linear exponential family -- which is why we show the mean function to
have the index form. If $Y$ is a binary or fractional response then we couple the
Bernoulli QLL with the logistic mean function -- just as when we estimate $%
\rho _{0} (x)$ and $\rho _{1} (x)$. If $Y\geq 0$,
the appropriate combination is $m\left( \cdot \right) =\exp \left( \cdot
\right) $ with the Poisson QLL\@. When $Y$ has no particular features worth
exploiting, one commonly uses $m\left( \alpha +x\beta \right) =\alpha
+x\beta $ and the least squares objective function (which corresponds to the
normal QLL)\@.

As in the case of estimating the parametric models for $\rho _{0} (x)$ and $\rho _{1} (x)$, the objective functions
for estimating $\left( \alpha _{0},\beta _{0}\right) $ and $\left( \alpha
_{1},\beta _{1}\right) $ are weighted by $1/\left[ 1-G(X_{i},\hat{\gamma})%
\right] $ and $1/G(X_{i},\hat{\gamma})$ for the $Z_{i}=0$ and $Z_{i}=1$
subsamples, respectively. When the mean function corresponds to the
canonical link function in the chosen linear exponential family (LEF), the
first-order conditions for $(\hat{\alpha}_{1},\hat{\beta}_{1})$ can be
written as
\begin{eqnarray}
\sum_{i=1}^{N}\left\{ \frac{Z_{i}}{G(X_{i},\hat{\gamma})}\left[ Y_{i}-m(\hat{%
\alpha}_{1}+X_{i}\hat{\beta}_{1})\right] \right\} &=&0 \label{sample_FOC1}
\\
\sum_{i=1}^{N}\left\{ \frac{Z_{i}}{G(X_{i},\hat{\gamma})}X_{i}^{\prime }%
\left[ Y_{i}-m(\hat{\alpha}_{1}+X_{i}\hat{\beta}_{1})\right] \right\} &=&0.
\label{sample_FOC2}
\end{eqnarray}%
Under general conditions, $(\hat{\alpha}_{1},\hat{\beta}_{1})$ converge in
probability to the (unique) solutions $(\alpha _{1}^{\ast },\beta _{1}^{\ast
})$ to the weighted population moment conditions
\begin{eqnarray}
\mathbb{E}\left\{ \frac{Z}{G(X,\gamma ^{\ast })}\left[ Y-m(\alpha _{1}^{\ast
}+X\beta _{1}^{\ast })\right] \right\} &=&0  \label{FOC1} \\
\mathbb{E}\left\{ \frac{Z}{G(X,\gamma ^{\ast })}X^{\prime }\left[ Y-m(\alpha
_{1}^{\ast }+X\beta _{1}^{\ast })\right] \right\} &=&0.  \label{FOC2}
\end{eqnarray}%
We now show that the solutions to these FOCs result in doubly robust
estimators of $\theta _{0}$ and $\theta _{1}$. We show the argument for the
latter with an almost identical argument for $\theta _{0}$.

As discussed in \cite{Wooldridge2007}, when the weights depend on conditioning
variables -- in this case, $X$ -- and the relevant feature of the
conditional distribution is correctly specified -- in this case, the
conditional mean $\mu _{1}(X)\equiv \mathbb{E}\left( Y|X,Z=1\right) $ --
weighting a suitably chosen objective function does not alter consistency of
the estimators. We can see this directly from the population FOCs. Assume
there are values $\left( \alpha _{1}^{\ast },\beta _{1}^{\ast }\right) $
such that
\begin{equation*}
\mathbb{E}\left( Y|X,Z=1\right) =m(\alpha _{1}^{\ast }+X\beta _{1}^{\ast }),
\end{equation*}%
so that the conditional mean is correctly specified. Then $ZY=Zm(\alpha
_{1}^{\ast }+X\beta _{1}^{\ast })+ZU$ and, since $\mathbb{E}\left(
ZU|X,Z\right) =0$, it follows immediately that
\begin{equation*}
\mathbb{E}\left\{ Z\left[ Y-m(\alpha _{1}^{\ast }+X\beta _{1}^{\ast })\right]
|X\right\} =0.
\end{equation*}%
Because $G(X,\gamma ^{\ast })>0$ is a function of $X$, it follows that
\begin{equation*}
\mathbb{E}\left\{ \left. \frac{Z}{G(X,\gamma ^{\ast })}\left[ Y-m(\alpha
_{1}^{\ast }+X\beta _{1}^{\ast })\right] \right\vert X\right\} =0.
\end{equation*}%
Given $G(X,\gamma ^{\ast })>0$ and sufficient variability in $X$ when $Z=1$,
the solutions to (\ref{FOC1}) and (\ref{FOC2}), $\left( \alpha _{1}^{\ast
},\beta _{1}^{\ast }\right) $, are unique. By iterated expectations,
\begin{equation*}
\theta _{1}=\mathbb{E}\left[ m(\alpha _{1}^{\ast }+X\beta _{1}^{\ast })%
\right].
\end{equation*}%
Similarly, $\theta _{0}=\mathbb{E}[m(\alpha _{0}^{\ast }+X\beta _{0}^{\ast
})]$ when $\mathbb{E}\left( Y|X,Z=0\right) =m(\alpha _{0}^{\ast }+X\beta
_{0}^{\ast })$. This is the first half of the double robustness result,
which does not actually use the assumption of a canonical link in the linear
exponential family.

For the other part of DR, it is useful to express the population FOCs
somewhat differently. Plug in for $Y$ and use $ZY=Z\mu _{1}(X)+ZU$ to get
\begin{eqnarray*}
\mathbb{E}\left[ \frac{Z}{G(X,\gamma ^{\ast })}[\mu _{1}(X)+U-m(\alpha
_{1}^{\ast }+X\beta _{1}^{\ast })]\right] &=&0 \\
\mathbb{E}\left[ \frac{Z}{G(X,\gamma ^{\ast })}X^{\prime }[\mu
_{1}(X)+U-m(\alpha _{1}^{\ast }+X\beta _{1}^{\ast })]\right] &=&0
\end{eqnarray*}%
or, because $\mathbb{E}\left( U|X,Z\right) =0$,
\begingroup
\allowdisplaybreaks
\begin{eqnarray}
\mathbb{E}\left[ \frac{Z}{G(X,\gamma ^{\ast })}[\mu _{1}(X)-m(\alpha
_{1}^{\ast }+X\beta _{1}^{\ast })]\right] &=&0  \label{FOC1_noY} \\
\mathbb{E}\left[ \frac{Z}{G(X,\gamma ^{\ast })}X^{\prime }[\mu
_{1}(X)-m(\alpha _{1}^{\ast }+X\beta _{1}^{\ast })]\right] &=&0.
\label{FOC2_noY}
\end{eqnarray}
\endgroup
By iterated expectations, these equations are equivalent to
\begin{eqnarray}
\mathbb{E}\left[ \frac{\eta (X)}{G(X,\gamma ^{\ast })}[\mu _{1}(X)-m(\alpha
_{1}^{\ast }+X\beta _{1}^{\ast })]\right] &=&0  \label{FOC1_IVPS} \\
\mathbb{E}\left[ \frac{\eta (X)}{G(X,\gamma ^{\ast })}X^{\prime }[\mu
_{1}(X)-m(\alpha _{1}^{\ast }+X\beta _{1}^{\ast })]\right] &=&0.
\label{FOC2_IVPS}
\end{eqnarray}%
When $G(x,\gamma )$ is correctly specified, $\eta (X)=G(X,\gamma ^{\ast })$,
and the first population moment condition becomes
\begin{equation*}
\mathbb{E}\left[ \mu _{1}(X)-m(\alpha _{1}^{\ast }+X\beta _{1}^{\ast })%
\right] =0.
\end{equation*}%
It follows immediately that $\theta _{1}=\mathbb{E}[m(\alpha _{1}^{\ast
}+X\beta _{1}^{\ast })]$, even though the conditional mean function need
not be correctly specified. This part of the DR\ result uses the assumption
that we have chosen the canonical link function for the chosen LEF\ density.
The same argument holds for $\theta _{0}$. Except for adding standard
regularity conditions, we have established consistency of the DR estimators
that combine inverse probability weighting (IPW) and regression adjustment
(RA), where RA is defined generally to include QMLEs in the LEF with a
canonical link function.

Because the LEF/canonical link combinations play an important role in DR
estimation, we summarize the common choices for the quasi-likelihoods and
mean functions in Table \ref{table1}.

\begin{table}[!tb]
\begin{adjustwidth}{-1.25in}{-1.25in}
\centering
\begin{threeparttable}
\caption{Combinations of QLLFs and Canonical Link Functions\label{table1}}
\begin{tabular}{ccc}
\hline\hline
\textbf{Support Restrictions} & \textbf{Mean Function} & \textbf{Quasi-LLF} \\
\hline
None & Linear & Gaussian \\ 
$Y(w)\in \left[ 0,1\right] $ (binary, fractional) & Logistic & Bernoulli \\ 
$Y(w)\in \left[ 0,B\right] $ (count, corner) & Logistic & Binomial \\ 
$Y(w)\geq 0$ (count, continuous, corner) & Exponential & Poisson \\
\hline
\end{tabular}
\end{threeparttable}
\end{adjustwidth}
\end{table}

The first entry in Table \ref{table1} simply means using weighted least squares with
linear conditional mean functions, but the weights here are based on the
instrument propensity score, chosen to achieve double robustness, and have
nothing to do with heteroskedasticity. When $Y$ is binary or fractional (the second entry), a
logistic conditional mean function is more attractive because it ensures
fitted values are in the unit interval. For example, $Y_{i}$ could be the
fraction of retirement savings held in the stock market or the fraction of
students passing a standardized test. The third entry allows for corners at
zero and some unit-specific, known upper bound, $B_{i}$. This $B_{i}$ should
be a conditioning variable -- like the elements of $X_{i}$. For example, $%
Y_{i}$ could be the amount of income put into retirement with $B_{i}$ being
an individual-specific bound determined by legal restrictions. The final
entry is important across many kinds of response variables that are
nonnegative but have no natural upper bound. These outcomes could be count
variables but they could be roughly continuous or have an atom at zero.

In what follows, we summarize the steps for doubly robust estimation of $\tau _{LATE}$ using
IPWRA\@.

\begin{uprocedure}{DR LATE}
\begin{enumerate}
\item[]
\item Estimate a flexible binary response model for the instrument propensity
score, $\eta \left( x\right) =\mathbb{P}\left( Z=1|X=x\right) $; denote the
fitted probabilities $G\left( X_{i},\hat{\gamma}\right) $. In many cases,
one would use a flexible logit. Overlap needs to be studied at this step.
\item Use weighted Bernoulli QMLE to estimate separate (flexible) logit models
for $\mathbb{P}(W=1|X,Z=0)$ and $\mathbb{P}(W=1|X,Z=1)$ (i.e., only using the units with $Z_{i}=0$ and $Z_{i}=1$, respectively), where the weights
in the former case are $1/\left[ 1-G\left( X_{i},\hat{\gamma}\right) \right] 
$, and in the latter case, $1/G\left( X_{i},\hat{\gamma}\right) $. These
produce $\left( \hat{\omega}_{0},\hat{\delta}_{0}\right) $, $\left( \hat{%
\omega}_{1},\hat{\delta}_{1}\right) $, and the fitted probabilities $\Lambda ( \hat{\omega}_{0}+X_{i}\hat{\delta}_{0} )$ and $\Lambda ( \hat{\omega}_{1}+X_{i}\hat{\delta}_{1} )$.
\item Choose conditional mean models for $\mathbb{E}(Y|X,Z=0)$ and $\mathbb{E}%
(Y|X,Z=1)$ that reflect the nature of $Y$\@. These should correspond to the
canonical link functions for the chosen LEF quasi-log likelihood. Use
weights $1/\left[ 1-G\left( X_{i},\hat{\gamma}\right) \right] $ to obtain
the weighted QMLEs of $\left( \alpha _{0},\beta _{0}\right) $ and weights $%
1/G\left( X_{i},\hat{\gamma}\right) $ to obtain the weighted QMLEs of $%
\left( \alpha _{1},\beta _{1}\right) $. (As above, this only uses the units with $Z_{i}=0$ and then $Z_{i}=1$, respectively.) These produce the fitted mean
functions $m ( \hat{\alpha}_{0}+X_{i}\hat{\beta}_{0} )$ and $m ( \hat{\alpha}_{1}+X_{i}\hat{\beta}_{1} )$.
\item Obtain the DR estimator of $\tau _{LATE}$ as
\begin{equation}
\hat{\tau}_{DRLATE}=\frac{N^{-1}\sum_{i=1}^{N}\left[ m ( \hat{\alpha}%
_{1}+X_{i}\hat{\beta}_{1} ) -m ( \hat{\alpha}_{0}+X_{i}\hat{\beta}%
_{0} ) \right] }{N^{-1}\sum_{i=1}^{N}\left[ \Lambda ( \hat{\omega}%
_{1}+X_{i}\hat{\delta}_{1} ) -\Lambda ( \hat{\omega}_{0}+X_{i}\hat{%
\delta}_{0} ) \right] }.\text{ }\square  \label{DRLATE}
\end{equation}%
\end{enumerate}
\end{uprocedure}

\noindent
The DR LATE estimator has the same form as \cite{Frolich2007}, but we use
parametric models that can exploit the nature of $Y$ and we estimate the
parameters in the numerator and denominator using inverse probability
weighting in order to achieve double robustness. Consequently, the numerator
of $\hat{\tau}_{DRLATE}$ is a DR ATE estimator where $Z$ (the instrument) is
taken as the \textquotedblleft treatment\textquotedblright\ and the outcome $%
Y$ is the response. The denominator is a DR ATE estimator where, again, $Z$
is the \textquotedblleft treatment\textquotedblright\ and the actual
treatment indicator, $W$, is the response. Obtaining the estimate for a given
sample is very easy using software packages that support IPWRA estimation.

\subsection{Estimation of LATT}

The IPWRA doubly robust estimators of $\tau _{LATT}$ require a different
weighting scheme. First, there is no need to model $\mu _{1}(X)=\mathbb{E}%
( Y|X,Z=1 ) $ or $\rho _{1}(X)=\mathbb{E} ( W|X,Z=1 )$
because, as shown in (\ref{lattcondid2}), we only need to estimate $%
\mathbb{E} ( Y|Z=1 )$ and $\mathbb{E} ( W|Z=1 )$. But
we need DR estimators of $\mathbb{E}[\mu _{0}(X)|Z=1]$ and $\mathbb{E}[\rho
_{0}(X)|Z=1]$. Following, for example, \cite{SW2018}, we now show that the following population FOC provides DR\
estimators of $\mathbb{E}[\mu _{0}(X)|Z=1]$:
\begin{equation}
\mathbb{E}\left\{ \frac{\left( 1-Z\right) G(X,\gamma ^{\ast })}{\left[
1-G(X,\gamma ^{\ast })\right] }\left[ Y-m(\alpha _{0}^{\ast }+X\beta
_{0}^{\ast })\right] \right\} =0.  \label{FOC_latt_1}
\end{equation}%
Using the same argument as for $\tau _{LATE}$, $\mathbb{E}\left(
Y|X,Z=0\right) =m(\alpha _{0}^{\ast }+X\beta _{0}^{\ast })$ ensures that
\begin{equation*}
\mathbb{E}\left\{ \left( 1-Z\right) \left[ Y-m(\alpha _{0}^{\ast }+X\beta
_{0}^{\ast })\right] |X\right\} =0,
\end{equation*}%
and then (\ref{FOC_latt_1}) holds by iterated expectations. Again, the
weights are nonnegative functions of $X$ and so this does not change that
the solutions to the FOCs are the conditional mean parameters.

For the second half of DR, we use an argument similar to the case of LATE and
write the FOC\ as
\begin{equation*}
\mathbb{E}\left\{ \frac{\left( 1-Z\right) G(X,\gamma ^{\ast })}{\left[
1-G(X,\gamma ^{\ast })\right] }\left[ \mu _{0}\left( X\right) -m(\alpha
_{0}^{\ast }+X\beta _{0}^{\ast })\right] \right\} =0.
\end{equation*}%
By iterated expectations, this FOC is equivalent to
\begin{equation*}
\mathbb{E}\left\{ \frac{\left[ 1-\eta \left( X\right) \right] G(X,\gamma
^{\ast })}{\left[ 1-G(X,\gamma ^{\ast })\right] }\left[ \mu _{0}\left(
X\right) -m(\alpha _{0}^{\ast }+X\beta _{0}^{\ast })\right] \right\} =0,
\end{equation*}%
where $\eta \left( X\right) =\mathbb{P}\left( Z=1|X\right) $. When the
instrument propensity score is correctly specified, $\eta \left( X\right)
=G(X,\gamma ^{\ast })$, this equation becomes
\begin{equation*}
\mathbb{E}\left\{ \eta \left( X\right) \left[ \mu _{0}\left( X\right)
-m(\alpha _{0}^{\ast }+X\beta _{0}^{\ast })\right] \right\} =0,
\end{equation*}%
and by iterated expectations, this is equivalent to
\begin{equation*}
\mathbb{E}\left\{ Z\left[ \mu _{0}\left( X\right) -m(\alpha _{0}^{\ast
}+X\beta _{0}^{\ast })\right] \right\} =0.
\end{equation*}%
It now follows that
\begin{equation*}
\mathbb{E}\left[ \mu _{0}\left( X\right) |Z=1\right]=\mathbb{E}\left[ m(\alpha _{0}^{\ast
}+X\beta _{0}^{\ast })|Z=1\right]
\end{equation*}%
even though the mean function is arbitrarily misspecified. This is the
second half of the DR result for $\tau _{LATT}$.

\begin{uprocedure}{DR LATT}
\begin{enumerate}
\item[]
\item Using all of the data, estimate a flexible binary response model for
the instrument propensity score, $\eta \left( x\right) =\mathbb{P}\left(
Z=1|X=x\right) $; denote the fitted probabilities $G\left( X_{i},\hat{\gamma}%
\right) $. In many cases, one would use a flexible logit. The LATT overlap
assumption needs to be studied at this step.
\item Use the units with $Z_{i}=0$ and weighted Bernoulli QMLE to estimate a
(flexible) logit model for $\rho_{0} (X) = \mathbb{P}(W=1|X,Z=0)$, where the weights are $%
G(X_{i},\hat{\gamma})/\left[ 1-G\left( X_{i},\hat{\gamma}\right) \right] $.
This produces $\left( \hat{\omega}_{0},\hat{\delta}_{0}\right) $ and the
fitted probabilities $\Lambda ( \hat{\omega}_{0}+X_{i}\hat{\delta}%
_{0} )$.
\item Choose a conditional mean model for $\mu_{0} (X) = \mathbb{E}(Y|X,Z=0)$ that reflects
the nature of $Y$\@. This should correspond to the canonical link function
for the chosen LEF quasi-log likelihood. Use the units with $Z_{i}=0$ and weights $G(X_{i},\hat{\gamma})/%
\left[ 1-G\left( X_{i},\hat{\gamma}\right) \right] $ to obtain the weighted
QMLEs of $\left( \alpha _{0},\beta _{0}\right) $. This produces the fitted mean, $m ( \hat{\alpha}_{0}+X_{i}\hat{\beta}_{0} )$.
\item Obtain the DR estimator of $\tau _{LATT}$ as
\begin{equation}
\hat{\tau}_{DRLATT}=\frac{ \bar{Y}_{1}-N_{1}^{-1}\sum_{i=1}^{N}Z_{i} m%
( \hat{\alpha}_{0}+X_{i}\hat{\beta}_{0} ) }{ \bar{W}%
_{1}-N_{1}^{-1}\sum_{i=1}^{N}Z_{i} \Lambda ( \hat{\omega}_{0}+X_{i}\hat{%
\delta}_{0} ) },  \label{DRLATT}
\end{equation}%
where $\bar{Y}_{1}=N_{1}^{-1}\sum_{i=1}^{N}Z_{i}Y_{i}$ and $\bar{W}%
_{1}=N_{1}^{-1}\sum_{i=1}^{N}Z_{i}W_{i}$. $\square $
\end{enumerate}
\end{uprocedure}

\noindent
The numerator of $\hat{\tau}_{DRLATT}$ is a DR\ estimator of ATT where $%
Z $ is taken as the \textquotedblleft treatment\textquotedblright\ and $Y$
is the outcome. Similarly, the denominator is a DR\ estimator of ATT
where $Z$ is taken as the treatment and the actual treatment status, $W$, is
the outcome.

\section{Inference}

\label{inference}

To perform valid inference on $\tau _{LATE}$ and $\tau _{LATT}$, such as
obtaining confidence intervals, we need to obtain standard errors for $%
\hat{\tau}_{DRLATE}$ and $\hat{\tau}_{DRLATT}$ that account for the sampling
error in all of the estimators and also the sample averages in (\ref{DRLATE}%
) and (\ref{DRLATT}). One possibility is to use a resampling scheme. The
most convenient is the nonparametric bootstrap, which resamples all
variables (and so accounts for sampling error in the estimators and in the
averages). Given a bootstrapped standard error we can easily obtain
asymptotically valid confidence intervals for $\tau _{LATE}$ and $\tau
_{LATT}$.

Because bootstrapping is not always desirable, we summarize a method of
obtaining a valid standard error that stacks the first-order conditions for
all estimation problems and then obtains a proper standard error from the
resulting generalized method of moments framework. We explicitly consider
how to do this for $\hat{\tau}_{DRLATE}$.

To allow one to choose the treatment binary response models and the
conditional mean models in a way that does not (theoretically) lead to DR\
estimation, let $m_{0}(X,\alpha _{0},\beta _{0})$ and $m_{1}(X,\alpha
_{1},\beta _{1})$ be the parametric models for $\mu _{0}(X)$ and $\mu
_{1}(X) $, respectively, and let $p_{0}(X,\omega _{0},\delta _{0})$ and $%
p_{1}(X,\omega _{1},\delta _{1})$ be the parametric models for $\rho _{0}(X)$
and $\rho _{1}(X)$, respectively. $G(X,\gamma )$ is the parametric model for 
$\mathbb{P}\left( Z=1|X\right) $. Let $\tau _{Y|Z}$ and $\tau _{W|Z}$ be the
numerator and denominator of the LATE, respectively, that is, 
\begingroup
\allowdisplaybreaks
\begin{eqnarray*}
\tau _{Y|Z} &=&\theta _{1}-\theta _{0}, \\
\tau _{W|Z} &=&\pi _{1}-\pi _{0}.
\end{eqnarray*}
\endgroup
Let $\phi =(\alpha _{0},\beta _{0},\alpha _{1},\beta _{1},\omega _{0},\delta
_{0},\omega _{1},\delta _{1},\gamma ,\tau _{Y|Z},\tau _{W|Z})$ and $%
S=(Y,X,W,Z)$. The estimators can be defined as a solution for the following
sample moment equation: 
\begin{equation}
\sum_{i=1}^{N}\psi (S_{i},\hat{\phi})=0.  \label{M-est}
\end{equation}%
By standard results for estimators that solve a first-order condition, it
follows that: 
\begin{equation}
\sqrt{N}(\hat{\phi}-\phi )\overset{a}{\sim }\mathrm{Normal}\left(
0,A^{-1}VA^{-1}\right) ,  \label{mvar2}
\end{equation}%
where 
\begin{eqnarray*}
A &\equiv &\mathbb{E}\left[ \frac{\partial \psi (S_{i},\phi )}{\partial \phi
^{\prime }}\right], \\
V &\equiv &\mathbb{V}[\psi (S_{i},\phi )]=\mathbb{E}[\psi (S_{i},\phi )\psi
(S_{i},\phi )^{\prime }].
\end{eqnarray*}%
Using the moment functions related to each parameter, the moment function $%
\psi (S_{i},\hat{\phi})$ in equation (\ref{M-est}) can be written explicitly
in the following way: 
\begin{equation*}
\psi (S_{i},\phi )=\left( 
\begin{array}{c}
\psi _{1}(S_{i},\phi ) \\ 
\psi _{2}(S_{i},\phi ) \\ 
\psi _{3}(S_{i},\phi ) \\ 
\psi _{4}(S_{i},\phi ) \\ 
\psi _{5}(S_{i},\phi ) \\ 
\psi _{6}(S_{i},\phi ) \\ 
\psi _{7}(S_{i},\phi )%
\end{array}%
\right) =\left( 
\begin{array}{c}
\frac{Z_{i}}{G(X_{i},\gamma )}\frac{\partial q_{1}^{y}(Y_{i},X_{i}; \alpha
_{1},\beta _{1})}{\partial (\alpha _{1},\beta _{1})} \\ 
\frac{1-Z_{i}}{1-G(X_{i},\gamma )}\frac{\partial
q_{0}^{y}(Y_{i},X_{i}; \alpha _{0},\beta _{0})}{\partial (\alpha _{0},\beta
_{0})} \\ 
\frac{Z_{i}}{G(X_{i},\gamma )}\frac{\partial q_{1}^{w}(W_{i},X_{i}; \omega
_{1},\delta _{1})}{\partial (\omega _{1},\delta _{1})} \\ 
\frac{1-Z_{i}}{1-G(X_{i},\gamma )}\frac{\partial
q_{0}^{w}(W_{i},X_{i}; \omega _{0},\delta _{0})}{\partial (\omega _{0},\delta
_{0})} \\ 
\frac{Z_{i}-G(X_{i},\gamma )}{G(X_{i},\gamma )(1-G(X_{i},\gamma ))%
}\frac{\partial G(X_{i},\gamma )}{\partial \gamma } \\ 
m_{1}(X_{i},\alpha _{1},\beta _{1})-m_{0}(X_{i},\alpha _{0},\beta _{0})-\tau
_{Y|Z} \\ 
p_{1}(X_{i},\omega _{1},\delta _{1})-p_{0}(X_{i},\omega _{0},\delta
_{0})-\tau _{W|Z}%
\end{array}%
\right) ,
\end{equation*}%
where $q_{z}^{y}(\cdot )$ and $q_{z}^{w}(\cdot )$ are the objective
functions for the estimation problems involving $Y$ and $W$, respectively.
The moment condition $\psi _{1}(S_{i},\phi )$ corresponds to the FOCs given
in equations (\ref{FOC1}) and (\ref{FOC2}). Similarly, $\psi _{2}(S_{i},\phi )$
is the FOC for the estimation of $(\alpha _{0},\beta _{0})$, and so on. The
moment conditions $\psi _{6}(S_{i},\phi )$ and $\psi _{7}(S_{i},\phi )$
account for the sampling variation in $X_{i}$ in obtaining $\hat{\tau}_{Y|Z}$
and $\hat{\tau}_{W|Z}$.

The asymptotic distribution of any parametric LATE estimator that uses
consistent estimators of $\tau _{Y|Z}$ and $\tau _{W|Z}$ can be derived for
a known joint asymptotic distribution of these estimators $\hat{\tau}_{Y|Z}$
and $\hat{\tau}_{W|Z}$, which satisfy:%
\begin{equation*}
\sqrt{N}\left[ 
\begin{pmatrix}
\hat{\tau}_{Y|Z} \\ 
\hat{\tau}_{W|Z}%
\end{pmatrix}%
-%
\begin{pmatrix}
\tau _{Y|Z} \\ 
\tau _{W|Z}%
\end{pmatrix}%
\right] \overset{d}{\rightarrow }\mathrm{Normal}(0,\Omega ),
\end{equation*}%
where $\Omega $ is the $2\times 2$ variance-covariance matrix corresponding
to the lower right block of $A^{-1}VA^{-1}$. Given $\hat{\tau}_{LATE}=\hat{%
\tau}_{Y|Z}/\hat{\tau}_{W|Z}$, we can apply the delta method to obtain
\begin{eqnarray}
\mathbb{AVAR}\left( \hat{\tau}_{DRLATE}\right) &=&\left( \frac{1}{\tau
_{W|Z}^{2}}\right) \mathbb{AVAR}\left( \hat{\tau}_{Y|Z}\right) +\left( \frac{%
\tau _{Y|Z}}{(\tau _{W|Z})^{2}}\right) ^{2}\mathbb{AVAR}\left( \hat{\tau}%
_{Y|Z}\right)  \nonumber \\
&&-\left( \frac{2\tau _{Y|Z}}{(\tau _{W|Z})^{3}}\right) \mathbb{ACOV}\left( 
\hat{\tau}_{Y|Z},\hat{\tau}_{W|Z}\right).  \label{asymvar_late}
\end{eqnarray}%
The three asymptotic variance terms are available from $\hat{\Omega}/N$, and
the other terms are easily estimated by plugging in $\hat{\tau}_{Y|Z}$ and $%
\hat{\tau}_{W|Z}$.

\section{A Test Comparing LATT and ATT Estimators}

\label{testing}

In textbook treatment of instrumental variables, where the treatment effect
is taken to be constant, it is fairly common to construct a \cite{Hausman1978} test for comparing the IV\ estimator with the OLS estimator of the
coefficients on the endogenous explanatory variable, $W$. The idea is that,
with good controls in $X$, maybe $W$ is unconfounded conditional on $X$, and
the instrumental variables are not needed. In practice, with cross-sectional
data one uses a heteroskedasticity-robust version of the Hausman test that
is easily implemented using a control function regression; see, for example, \citet[Section 6.3.1]{Wooldridge2010}. In the traditional setting, efficiency
considerations are the primary reason for preferring OLS if the Hausman test
does not reject the null that $W$ is unconfounded: the OLS estimator is
typically much more precise than the IV\ estimator.

Efficiency remains a valid consideration when treatment effects are
heterogeneous, as in the current setting, but the usual Hausman test is no longer valid because OLS and IV estimands constitute different weighted averages of heterogeneous treatment effects even under the null. Instead, if $W$ is unconfounded conditional
on $X$ then, with sufficient overlap, one can identify the average treatment
effect on the treated (ATT) without requiring an instrumental variable. As
shown by \citetalias{DHL2014b}, under one-sided noncompliance, LATT
is the same as ATT\@. Therefore, it makes sense to use doubly robust
estimators of the ATT (DR ATT) that do not use an instrument and compare that
with the DR LATT estimates.

The ATT parameter is
\begin{equation}
\tau _{ATT}=\mathbb{E}\left[ Y(1) -Y(0) \mid W=1\right].
\label{ATT}
\end{equation}%
Following \cite{SW2018}, we use DR estimators of $%
\tau _{ATT}$ that are natural given the form of the DR LATT estimators in
Section \ref{robust}. We no longer need an instrument propensity score.
Instead, let $F\left( x,\gamma \right) $ be a model of the treatment
propensity score, $\mathbb{P}\left( W=1|X=x\right) $. Given a random sample
of size $N$, let $\hat{\gamma}$ be the (quasi-) MLE\ based on the Bernoulli
log likelihood. As before, a typical choice of $F\left( x,\gamma \right) $
is a flexible logistic function.

Under the assumption that $W$ is independent of $Y(0)$
conditional on $X$ and the overlap assumption
\begin{equation}
\mathbb{P}\left( W=1|X=x\right) <1\text{ for almost all }x\in \mathcal{X}\text{,}
\label{overlap_ATT}
\end{equation}%
we can obtain DR estimators of ATT using quasi-MLE in the LEF with a
canonical link function.

The conditional mean we need to estimate is $\mathbb{E}\left[ Y(0) |X=x\right] $, and we again take the model to have the index form, $%
m ( \alpha _{0}+x\beta _{0} )$ (reusing earlier notation). As
before, we can choose $m ( \alpha _{0}+x\beta _{0} )$ to reflect
the nature of the outcome variable $Y(0)$. To stay within the DR
class of estimators using IPWRA, $m ( \cdot )$ will be the
identity, logistic, or exponential function in the vast majority of
applications.

For consistent estimation of $\tau _{ATT}$ we can get by with the
conditional mean version of unconfoundedness of $W$ conditional on $X$,
\begin{equation*}
\mathbb{E}\left[ Y(0) |W,X\right] =\mathbb{E}\left[ Y\left(
0\right) |X\right] .
\end{equation*}%
Under this assumption, if the mean is correctly specified then $\alpha _{0}$
and $\beta _{0}$ are identified by
\begin{equation*}
\mathbb{E}\left( Y|W=0,X\right) =\mathbb{E}\left[ Y(0) |X\right]
=m ( \alpha _{0}+X\beta _{0} ).
\end{equation*}%
Letting $q\left( y,m\right) $ be the quasi-log likelihood function, $%
\hat{\alpha}_{0}$ and $\hat{\beta}_{0}$ solve the weighted QMLE problem
\begin{equation*}
\max_{a_{0}\mathbf{,}b_{0}}\sum_{i=1}^{N}\frac{F(X_{i}\hat{\gamma})}{1-F(X_{i}\hat{\gamma})}\left( 1-W_{i}\right) q\left( Y_{i},m (
a_{0}+X_{i}b_{0} ) \right) ,
\end{equation*}%
where the estimation is done on the control sample and the weighting ensures
the DR property; see \cite{SW2018}. Given the
estimates, the DR\ estimator of $\tau _{ATT}$ is
\begin{equation}
\hat{\tau}_{DRATT}=\bar{Y}_{1}-N_{1}^{-1}\sum_{i=1}^{N}W_{i}\cdot m ( \hat{\alpha}_{0}+X_{i}\hat{\beta}_{0} ),  \label{dr_att}
\end{equation}%
where $\bar{Y}_{1}$ is the average outcome over the treated units and $N_{1}$
is now the number of treated (not eligible) units. The estimator in (\ref%
{dr_att}) has a simple interpretation as an imputation estimator, as the
second term is a DR estimator of $\mathbb{E}\left[ Y(0) |W=1%
\right] $ obtained by first imputing $\mathbb{E}\left[ Y_{i}\left( 0\right)
|W_{i}=1,X_{i}\right] $ using the mean function estimated from the $W_{i}=0$
units. This DR estimator is pre-programmed in popular statistics and
econometrics packages.

Given $\hat{\tau}_{DRLATT}$ from Section \ref{robust} and $\hat{\tau}_{DRATT}$ in (\ref{dr_att}), we can test the null hypothesis that treatment is unconfounded given $X$,
provided the instrument $Z$ is such that one-sided noncompliance holds so
that $\tau _{LATT}=\tau _{ATT}$. A formal comparison is based on the
statistic
\begin{equation}
\frac{ \hat{\tau}_{DRLATT}-\hat{\tau}_{DRATT} }{\mathrm{se} \left( \hat{\tau}_{DRLATT}-\hat{\tau}_{DRATT} \right)}.  \label{latt_t_stat}
\end{equation}%
Under the null hypothesis, we cannot say that $\hat{\tau}_{DRATT}$ is the
efficient estimator in a suitable class that includes $\hat{\tau}_{DRLATT}$,
and so the standard error $\mathrm{se} \left( \hat{\tau}_{DRLATT}-\hat{\tau}_{DRATT} \right)$ does not simplify. Nevertheless, bootstrapping is
computationally feasible, or one can extend the calculations in Section \ref%
{inference} to obtain an analytical standard error. Even without one-sided noncompliance, a similar test can also be constructed to assess treatment effect heterogeneity by comparing DR LATE and DR LATT, or IV and DR LATE, or IV and DR LATT estimates.

\section{Empirical Applications}
\label{applications}

In this section we reanalyze the data in \cite{Abadie2003} and \cite{TAWBF2014} to illustrate our new doubly robust estimators. 

\subsection{The Effects of 401(k) Retirement Plans}
\label{abadie}

The 401(k) retirement plans were introduced in the US to increase saving for retirement by allowing tax advantages for the contributions to the retirement account. The policy-relevant empirical question is whether the 401(k) program is effective for increasing savings or only crowds out other personal saving. Since the individuals who participate in 401(k) plans are likely to have different saving preferences than non-participating individuals, a simple comparison of savings of the two groups is likely to provide an upward-biased estimate of the true effect. Different from other saving plans, 401(k) participation requires eligibility which, in turn, is determined by the employer. \cite{Abadie2003} argues, following \cite{PVW1994, PVW1995}, that 401(k) eligibility can be used as a conditionally independent instrument to estimate the effects of 401(k) participation on savings. Several recent papers have revisited this empirical question and estimated the LATE using the same instrument by  different methods \citep[e.g.,][]{BCFVH2017, Chernozhukovetal2018, Heiler2022, SASX2022}. Since our new doubly robust estimator relies on the same identifying assumptions, we also choose to reanalyze the effect of 401(k) participation. We use the same dataset as \cite{Abadie2003}. The data consists of a sample of 9,275 households from  the Survey of Income and Program Participation (SIPP) of 1991. As outcomes, we consider net financial assets (in US dollars) and participation in an individual retirement account (IRA), which is another popular tax-deferred saving plan in the US\@. The treatment is an indicator for participation in a 401(k) plan. The set of control variables consists of family income, age, marital status, and family size. Age enters the conditional mean functions quadratically.

\begin{table}[!tb]
\begin{adjustwidth}{-1.25in}{-1.25in}
\centering
\begin{threeparttable}
\caption{Estimates of the Effects of 401(k) Participation\label{tab:table2}}
\begin{tabular}{cccccc}
\hline\hline
	& \multicolumn{2}{c}{\textbf{(A)}} &       & \multicolumn{2}{c}{\textbf{(B)}} \\
	& \multicolumn{2}{c}{\textbf{Net financial assets}} &       & \multicolumn{2}{c}{\textbf{IRA}} \\
\cline{2-3}
\cline{5-6}
	& \textbf{Estimate} & \textbf{Std.~err.} &       & \textbf{Estimate} & \textbf{Std.~err.} \\
\hline
\textbf{OLS}	& 13,527 & (1,810) &       & 0.0569 & (0.0103) \\
\textbf{IV}	& 9,419  & (2,152) &       & 0.0274 & (0.0132) \\
	&       &       &       &       &  \\
	& \multicolumn{5}{c}{\textbf{Hausman test}}\\
\textbf{$\mathbf{H_0}$: OLS $=$ IV} &\multicolumn{2}{c}{$p = 0.0004$} &       & \multicolumn{2}{c}{$p = 0.0004$} \\
	&       &       &       &       &  \\
	& \multicolumn{2}{c}{\textbf{ATE}} &       & \multicolumn{2}{c}{\textbf{ATE}} \\
\textbf{IPWRA} & 10,767 & (1,772) &       & 0.0554 & (0.0096) \\
	&       &       &       &       &  \\
	& \multicolumn{2}{c}{\textbf{LATE}} &       & \multicolumn{2}{c}{\textbf{LATE}} \\
\textbf{IPW} & 3,994  & (4,891) &       & 0.0165 & (0.0135) \\
\textbf{RA} & 8,467  & (1,991) &       & 0.0338 & (0.0128) \\
\textbf{IPWRA} & 8,046  & (2,587) &       & 0.0361 & (0.0128) \\
\textbf{AIPW} & 5,416  & (4,176) &       & 0.0404 & (0.0131) \\
	&       &       &       &       &  \\
	& \multicolumn{2}{c}{\textbf{ATT}} &       & \multicolumn{2}{c}{\textbf{ATT}} \\
\textbf{IPWRA} & 12,673 & (3,329) &       & 0.0697 & (0.0110) \\
	&       &       &       &       &  \\
	& \multicolumn{2}{c}{\textbf{LATT}} &       & \multicolumn{2}{c}{\textbf{LATT}} \\
\textbf{IPWRA} & 10,918 &  (3,709)  &       & 0.0413 & (0.0143) \\
	&       &       &       &       &  \\
	& \multicolumn{5}{c}{\textbf{Hausman test}}\\
\textbf{$\mathbf{H_0}$: ATT $=$ LATT} & \multicolumn{2}{c}{$p = 0.457$} &       & \multicolumn{2}{c}{$p = 0.001$}  \\
\hline
\end{tabular}
\begin{footnotesize}
\begin{tablenotes}[flushleft]
\item \textit{Notes:} The data are \cite{Abadie2003}'s subsample of the Survey of Income and Program Participation (SIPP) of 1991. The sample size is 9,275. The outcomes are net financial assets (Panel A) and a binary indicator for participation in IRAs (Panel B)\@. The treatment is an indicator for 401(k) participation. The instrument is an indicator for 401(k) eligibility. The set of covariates consists of family income, age, age squared, marital status, and family size. ``OLS'' and ``IV'' are the estimates of the coefficient on the endogenous treatment with covariates (and instrument) listed above. The remaining estimators are defined in the main text. Standard errors are in parentheses. For OLS and IV, we report robust standard errors. For the remaining estimators, our standard errors follow from the GMM framework in Section \ref{inference}.
\end{tablenotes}
\end{footnotesize}
\end{threeparttable}
\end{adjustwidth}
\end{table}

Table \ref{tab:table2} reports the estimates of the parameters of interest together with asymptotic standard errors for both outcome variables. The first two rows display the coefficient estimates for 401(k) participation from OLS and IV estimation, respectively. OLS estimates of the participation coefficient for both regressions are positive and significant. However, as mentioned above, due to unobserved preferences for saving, it is likely that these overestimate the true effect even after conditioning on various individual characteristics. The usual IV estimates are indeed much smaller than the OLS estimates, although they are still positive and significant. Moreover, the Hausman test for the absence of endogeneity strongly rejects for both outcomes.

Next, we report the average treatment effect (ATE) of the participation in a 401(k) plan, as estimated by IPWRA\@. The ATE is identified if there are no unmeasured confounders. If this is not the case, similar to the OLS coefficient, we also expect the ATE estimates to be biased. The ATE estimate of 401(k) participation on total net financial assets is smaller than the OLS estimate and slightly larger than the IV estimate. The ATE on having an IRA account is also larger than the corresponding IV estimate but closer to the OLS estimate. Following the estimated ATEs, the LATE estimates are also reported. Specifically, we estimate the LATE using IPW, RA, AIPW, and our proposed IPWRA estimator. The LATE is identified if the assumptions discussed in Section \ref{identification} are satisfied. IPW and AIPW estimates of the LATE on net financial assets are small in magnitude and very imprecisely estimated. On the other hand, the RA and IPWRA estimates are closer to the IV estimates and relatively very precise. The IPW estimate of the LATE on IRA participation is insignificant while the estimates based on RA, IPWRA, and AIPW are significant and rather similar in magnitude. Additionally, we estimate the ATT and LATT of participation in a 401(k) plan. Since there is one-sided noncompliance, we can test the equality of ATT and LATT as discussed in Section \ref{testing}. For net financial assets, we cannot reject the equality of the two treatment effects, ATT and LATT\@. This suggests that participation in a 401(k) plan might be unconfounded conditional on the set of variables that we control for, following \cite{Abadie2003}. However, the ATT and LATT for the probability of having an IRA are statistically different from each other based on our proposed test, which underscores the importance of IV estimation under the maintained assumptions of Section \ref{identification}.

The results of this application are reassuring given that our proposed DR method provides estimates in a reasonable range with good precision. At the same time, the AIPW estimate of the LATE on net financial assets -- that is, the other doubly robust estimate -- is very imprecise, which is in line with previous criticism of AIPW estimation in other contexts \citep{KS2007}. The RA estimate of the LATE has a smaller standard error but does not enjoy the double robustness property. The differences between the LATE estimates for the binary outcome, IRA participation, are minor, with the exception of the small and insignificant IPW estimate.

\subsection{The Effects of Medicaid}
\label{medicaid}

In our second empirical application, we revisit \cite{TAWBF2014}'s analysis of the data from the Oregon Health Insurance Experiment. In 2008, the state of Oregon decided to offer about 10,000 spots in Medicaid using a lottery that randomly selected eligible households from a larger pool of applicants. \cite{Finkelsteinetal2012}, \cite{TAWBF2014}, \cite{DL2022}, and \cite{JCK2022}, among others, used Oregon's lottery assignment as a binary instrument to assess the effects of Medicaid on various outcomes related to health and healthcare utilization. Here, following much of the previous work, we focus on emergency room (ER) visits at the extensive and intensive margins. In other words, we use our proposed method to analyze the effects of Medicaid on a binary outcome indicating any ER visits in the study period as well as on a count outcome indicating the number of visits.

In what follows, we use \cite{TAWBF2014}'s administrative data with over 24,000 observational units. The endogenous treatment variable is defined as ``ever enrolled in Medicaid'' during the study period. Table \ref{tab:table3} reports a number of estimates together with their associated standard errors. OLS refers to the coefficient estimates for the endogenous treatment variable. ATE and ATT refer to the IPWRA estimates of these parameters. IV refers to the coefficient estimates for the endogenous treatment variable, which use the binary lottery instrument. Finally, LATE and LATT are estimated using our proposed estimators. Following \cite{TAWBF2014}, all regressions include indicators for different numbers of household members on the lottery list, which is necessary for instrument validity, and past outcome data, which should improve the precision of the final estimates.

\begin{table}[!tb]
\begin{adjustwidth}{-1.25in}{-1.25in}
\centering
\begin{threeparttable}
\caption{Estimates of the Effects of Medicaid\label{tab:table3}}
\begin{tabular}{cccccccc}
\hline\hline
	& \textbf{OLS} & \textbf{ATE} & \textbf{ATT} &       & \textbf{IV} & \textbf{LATE} & \textbf{LATT} \\
\hline
\multicolumn{8}{c}{\textbf{Outcome: ER visits (any)}} \\
\textbf{Estimate} & 0.1355 & 0.1348 & 0.1360 &       & 0.0697 & 0.0696 & 0.0812 \\
\textbf{Std.~err.} & (0.0068) & (0.0069) & (0.0069) &       & (0.0239) & (0.0238) & (0.0246) \\
	&       &       &       &       &       &       &  \\
\multicolumn{8}{c}{\textbf{Outcome: ER visits (\#)}} \\
\textbf{Estimate} & 0.4611 & 0.5320 & 0.5622 &       & 0.3880 & 0.4508 & 0.4701 \\
\textbf{Std.~err.} & (0.0340) & (0.0346) & (0.0385) &       & (0.1070) & (0.1254) & (0.1141) \\
\hline
\end{tabular}
\begin{footnotesize}
\begin{tablenotes}[flushleft]
\item \textit{Notes:} The data are \cite{TAWBF2014}'s sample from the Oregon Health Insurance Experiment. The sample sizes are 24,646 (top panel) and 24,615 (bottom panel). The outcomes are an indicator for any ER visits (top panel) and the (censored) number of ER visits (bottom panel) in the study period. The treatment is an indicator for Medicaid coverage. The instrument is an indicator for whether a given household was selected by the Medicaid lottery. The set of covariates consists of indicators for different numbers of household members on the lottery list (both panels), an indicator for any ER visits before the randomization (top panel), and the number of ER visits before the randomization (bottom panel). ``OLS'' and ``IV'' are the estimates of the coefficient on the endogenous treatment with covariates (and instrument) listed above. The remaining estimators are defined in the main text. Standard errors, clustered on the household identifier, are in parentheses. For OLS and IV, we report cluster-robust standard errors. For the remaining estimators, our standard errors follow from the GMM framework in Section \ref{inference}.
\end{tablenotes}
\end{footnotesize}
\end{threeparttable}
\end{adjustwidth}
\end{table}

It turns out that the OLS, ATE, and ATT estimates (and their standard errors) are almost identical for the binary outcome and quite similar for the count outcome, although in the latter case the OLS estimates are smaller than those of the ATE and ATT\@. In any case, the estimates show a significant positive correlation between Medicaid and ER utilization both on the extensive and intensive margins. Due to treatment endogeneity, however, these estimates cannot be interpreted as causal effects. The last three columns of Table \ref{tab:table3} take this endogeneity into account and rely on the instrumental variable for the identification of causal effects. It turns out that these estimates are much smaller than the OLS, ATE, and ATT estimates, especially in the case of the binary outcome, where the estimates are now roughly half as large. Like in \cite{TAWBF2014}, and unlike in an earlier analysis by \cite{Finkelsteinetal2012} that relied on mail survey data, these estimates are also significantly different from zero. Our analysis, however, also reveals an interesting dimension of treatment effect heterogeneity: LATT, the effect on the treated compliers, appears to be larger than the usual LATE\@. A formal comparison of the two objects yields a \emph{p}-value of about 0.04 for the binary outcome and 0.47 for the count outcome. The fact that the LATT may be larger than the LATE is likely due to treatment effect heterogeneity across households of different sizes: effects of Medicaid on ER utilization are more pronounced in larger households \citep[cf.][]{DL2022}, which are also more likely to be treated given the lottery design.

\section{Simulations}
\label{simulations}

In this section, we conduct a Monte Carlo study to assess the bias and precision of our DR LATE estimator in comparison with other existing estimators. In particular, we focus on the effect of certain types of misspecification on the bias and precision. To eliminate to some extent the arbitrariness in choosing the data-generating process, we generate our Monte Carlo samples to mimic some statistical features of the 401(k) dataset we used in Section \ref{abadie}. We draw 1,000 samples with $N=1{,}000$ and the same number of samples with $N=4{,}000$ observations. Each sample is constructed in the following steps. First, we draw two random variables from a bivariate normal distribution. The parameters of the bivariate normal distribution are set equal to the empirical means and covariances of age and log income in the 401(k) data. The simulated log income is then exponentiated to generate the income variable. As an additional covariate, we take the square of the simulated age variable. Thus, our full set of covariates, $X$, includes three variables: income, age, and age squared. The instrumental variable $Z$ is generated according to
\begin{equation}
	Z = \1  \left(\Lambda \left( 
	\gamma_0 + X\gamma_x\right) > U_z \right), \label{dgp:z}
\end{equation} 
where $\gamma=(\gamma_0,\gamma_x)$ corresponds to the estimated coefficient vector from a logit regression of 401(k) eligibility on a constant, income, age, and age squared using the original data (see column (1) of Table \ref{tab1} in the Appendix). The random variable $U_z$ is drawn from the standard uniform distribution and $\Lambda(\cdot)$ is the logistic cdf. Then, we construct $D(1)$ as follows:
\begin{equation}
	D(1) = \1  \left(\Lambda \left( \omega_{0}+X\delta_{0}\right) > U_1 \right), \label{dgp:d1}
\end{equation} 
where the coefficients are from a logit regression using observations with $Z_{i}=1$ (column (2) of Table \ref{tab1} in the Appendix) in the original data and $ U_1$ is drawn from the standard uniform distribution. Finally, we generate two outcome variables. One mimics the continuous outcome variable, net financial assets, and the other mimics the binary outcome variable, participation in IRA\@. The continuous outcome $Y(z)$ is generated using the following linear model:
\begin{equation}	
	Y(z) =  \alpha_{z}+X\beta_{z} + \varepsilon_z \quad \text{ for } z=0,1. \label{dgp:conty}
\end{equation} 
We use the coefficients from two separate regressions of the outcome variable on the set of covariates for $Z_i=1$  and $Z_i=0$ subsamples in the original data (columns (3) and (4) of Table \ref{tab1} in the  Appendix). The error terms $\varepsilon_z$ are drawn from a normal distribution with mean zero and variance $\sigma_z^2$, where $\sigma_z^2$ is the mean squared residual from the regression for $Z_i=z$. The binary outcome variable is constructed similarly to the potential treatment variable using logit link:
 \begin{equation}
 	Y(z) =  \1  \left( \Lambda \left(\alpha_{z}+X\beta_{0} \right) > U_y \right)  \quad \text{ for } z=0,1, \label{dgp:biny}
 \end{equation} 
with coefficients from two separate logistic regressions of the binary indicator of IRA participation on the set of covariates (columns (5) and (6) of Table \ref{tab1} in the Appendix) and $U_y$ drawn from the standard uniform distribution. For the simulated data, the ``true'' values of the LATE are \$8,816.5 for net financial assets and 0.036 for the probability of IRA participation.

For each simulated sample, in addition to our proposed method, we estimate the LATE using IV, RA, IPW, and AIPW\@. The RA estimator of the LATE is constructed similarly to our proposed estimator with the crucial difference that the objective functions are not weighted. The IPW and AIPW estimators of the LATE are the ratios of two IPW and AIPW estimators of ATEs of $Z$ on $Y$ and $W$, respectively. 

In general, the LATE estimators based on the identification result in \eqref{lateid} require estimation of four conditional means. However, as mentioned earlier, since in the original 401(k) data and in our simulation design $Z=0$ implies $D=0$, the second term in the denominator of \eqref{lateid} is zero and we do not need to estimate \eqref{rho_0}. Thus, for the RA approach, we estimate the two conditional means in the numerator of \eqref{lateid} by two separate linear (logistic) regressions of continuous (binary) $Y(z)$  using observations from subsamples with $Z_i=1$ and $Z_i=0$, respectively. Similarly, the first conditional mean in the denominator is estimated by a logistic regression for the subsample with $Z_i=1$.  On the other hand, the IPW approach requires estimating a binary response model for the instrument propensity score
defined in (\ref{iv_ps}). Thus, we estimate the instrument propensity score by a logistic regression. Our proposed IPWRA method requires the same conditional mean specifications as the RA approach and additionally the specification of \eqref{iv_ps} to construct the weights. For IPWRA, we estimate the conditional means in the numerator of \eqref{lateid} by two separate \textit{weighted} linear (logistic) regressions of continuous (binary) $Y(z)$  with the weight equal to the inverse of the estimated probability of being eligible or not, using the subsample with $Z_i=1$ or $Z_i=0$, respectively. Finally, AIPW requires the same set of model specifications as our proposed method, although the regressions are not weighted.

In our Monte Carlo study, we consider estimators \emph{(i)} when the required models are all correctly specified, \emph{(ii)} when models for \eqref{mu_0}--\eqref{rho_1}  are misspecified, and \emph{(iii)} when the model for \eqref{iv_ps} is misspecified. Correct specifications for these estimators mean that we use the correct set of covariates for all the regressions, namely simulated income, age, and age squared. Misspecification of a certain model means that the set of regressors does not include age squared.

Tables \ref{tab:simresults1} and \ref{tab:simresults2}  present the biases, root mean squared errors of the
LATE estimators, and the empirical coverage rates for nominal 95\% confidence intervals under the different model specifications for the dependent variables net financial assets and IRA participation, respectively.

Table \ref{tab:simresults1} suggests that, when the relevant models are correctly specified, that is, all the confounding factors are controlled for, the bias is highest for the linear IV estimates. This coincides with the fact that the IV estimand is not equal to the LATE in the case of a conditionally independent instrument \citep[e.g.,][]{Sloczynski2021, BBMT2022}. The RA estimator has the smallest bias and RMSE when all the models are correctly specified. For the smaller sample size, our proposed estimator has the second smallest bias and RMSE\@. For the larger sample size, the bias estimates are very close for IPWRA and AIPW but the precision of our estimator is better than that of IPW and AIPW\@.

The second block of Table \ref{tab:simresults1} presents the results for the first type of misspecification. The IV estimator becomes heavily biased when we do not control for age squared. The effect of omitting this variable when estimating the conditional mean functions in \eqref{mu_0}--\eqref{rho_1} is similar for the RA estimator, which is also very biased. The IPW estimator is not affected by this type of misspecification since it only requires the correct specification of the model for the instrument propensity score. The doubly robust estimators, IPWRA and AIPW, are not seriously affected by this type of misspecification either, as predicted by theory. Our proposed method has the smallest bias and RMSE for sample size $N=1{,}000$, and a slightly larger bias -- but still the smallest RMSE -- for sample size $N=4{,}000$.

\begin{table}[!tb]
\begin{adjustwidth}{-1.25in}{-1.25in}
\centering
\begin{threeparttable}
\caption{Simulation Results for the Continuous Outcome Variable\label{tab:simresults1}}
\begin{tabular}{cccccccccccc}
\hline\hline
& \multicolumn{3}{c}{\textbf{All Correct}} & & \multicolumn{3}{c}{\textbf{\eqref{mu_0}--\eqref{rho_1} misspecified}} & & \multicolumn{3}{c}{\textbf{\eqref{iv_ps} misspecified}} \\
\cline{2-4}
\cline{6-8}
\cline{10-12}
& \multicolumn{1}{c}{\textbf{Bias}} & \multicolumn{1}{c}{\textbf{RMSE}} & \multicolumn{1}{c}{\textbf{Cov.}} & & \multicolumn{1}{c}{\textbf{Bias}} & \multicolumn{1}{c}{\textbf{RMSE}} & \multicolumn{1}{c}{\textbf{Cov.}} & & \multicolumn{1}{c}{\textbf{Bias}} & \multicolumn{1}{c}{\textbf{RMSE}} & \multicolumn{1}{c}{\textbf{Cov.}} \\
\hline
\textbf{N=1{,}000} &       &       &       &       &       &       &       &       &       &  \\
\cline{1-1}
\textbf{IV} & 271.43 & 6,163.19 & 95.3  & & --1,544.41 & 6,395.32 & 94.5  & & 271.43 & 6,163.19 & 95.3 \\
\textbf{RA} & 127.69 & 6,169.94 & 95.5  & & --1,724.20 & 6,445.16 & 94.2  & & 127.69 & 6,169.94 & 95.5 \\
\textbf{IPW} & 162.49 & 6,958.60 & 95.8  & & 162.49 & 6,958.60 & 95.8  & & --1,549.48 & 7,033.47 & 94.1 \\
\textbf{IPWRA} & 159.24 & 6,300.47 & 95.4  & & 103.68 & 6,306.11 & 95.3  & & 140.70 & 6,258.49 & 95.3 \\
\textbf{AIPW} & 195.36 & 6,418.33 & 95.6  & & 170.13 & 6,439.62 & 95.4  & & 170.75 & 6,304.15 & 95.4 \\
&       &       &       &       &       &       &       &       &       &  \\
\textbf{N=4{,}000} &       &       &       &       &       &       &       &       &       &  \\
\cline{1-1}
\textbf{IV} & 114.57 & 3,097.13 & 94.8  & & --1,734.22 & 3,565.59 & 89.7  & & 114.57 & 3,097.13 & 94.8 \\
\textbf{RA} & --45.16 & 3,119.43 & 94.4  & & --1,907.94 & 3,662.62 & 89.2  & & --45.16 & 3,119.43 & 94.4 \\
\textbf{IPW} & --60.69 & 3,381.29 & 94.4  & & --60.69 & 3,381.29 & 94.4  & & --1738.43 & 3,782.54 & 91.2 \\
\textbf{IPWRA} & --74.71 & 3,155.63 & 94.8  & & --102.96 & 3,161.44 & 94.8  & & --69.52 & 3,152.04 & 94.6 \\
\textbf{AIPW} & --74.41 & 3,174.61 & 94.8  & & --95.49 & 3,183.18 & 94.8  & & --67.71 & 3,160.11 & 94.7 \\
\hline
\end{tabular}
\begin{footnotesize}
\begin{tablenotes}[flushleft]
\item \textit{Notes:} The details of the simulation design are provided in Section \ref{simulations}. Results are based on 1,000 replications. ``RMSE'' is the root mean squared error of an estimator. ``Cov.''~is the coverage rate for a nominal 95\% confidence interval. ``IV'' is the IV estimate of the coefficient on the endogenous treatment, controlling for $X$\@.  The remaining estimators are defined in the main text. To calculate the coverage rate, we use robust standard errors (IV) or standard errors that follow from the GMM framework in Section \ref{inference} (remaining estimators).
\end{tablenotes}
\end{footnotesize}
\end{threeparttable}
\end{adjustwidth}
\end{table}

Finally, we investigate the bias and RMSE for the case where the instrument propensity score in \eqref{iv_ps} is estimated without the squared term. As expected, the only estimator that is severely affected by this misspecification is IPW\@. The doubly robust methods continue to have reasonably small biases. IPWRA has the smaller bias and RMSE for $N=1{,}000$ and the smaller RMSE for $N=4{,}000$, as in other cases. The results demonstrate that the double robustness of our proposed method is achieved without significant sacrifices in terms of precision. Coverage rates are close to the nominal coverage rate for both DR estimators in all cases. For IV, RA, and IPW, coverage rates are sometimes substantially lower than 95\% when the estimators are otherwise biased.

\begin{table}[!tb]
\begin{adjustwidth}{-1.25in}{-1.25in}
\centering
\begin{threeparttable}
\caption{Simulation Results for the Binary Outcome Variable\label{tab:simresults2}}
\begin{tabular}{cccccccccccc}
\hline\hline
& \multicolumn{3}{c}{\textbf{All Correct}} & & \multicolumn{3}{c}{\textbf{\eqref{mu_0}--\eqref{rho_1} misspecified}} & & \multicolumn{3}{c}{\textbf{\eqref{iv_ps} misspecified}} \\
\cline{2-4}
\cline{6-8}
\cline{10-12}
& \multicolumn{1}{c}{\textbf{Bias}} & \multicolumn{1}{c}{\textbf{RMSE}} & \multicolumn{1}{c}{\textbf{Cov.}} & & \multicolumn{1}{c}{\textbf{Bias}} & \multicolumn{1}{c}{\textbf{RMSE}} & \multicolumn{1}{c}{\textbf{Cov.}} & & \multicolumn{1}{c}{\textbf{Bias}} & \multicolumn{1}{c}{\textbf{RMSE}} & \multicolumn{1}{c}{\textbf{Cov.}} \\
\hline
\textbf{N=1{,}000} &       &       &       &       &       &       &       &       &       &  \\
\cline{1-1}
\textbf{IV} & --0.0005 & 0.0421 & 93.5  & & --0.0048 & 0.0423 & 93.6  & & --0.0005 & 0.0421 & 93.5 \\
\textbf{RA} & --0.0031 & 0.0408 & 94.0    & & --0.0030 & 0.0407 & 93.3  & & --0.0031 & 0.0408 & 94.0 \\
\textbf{IPW} & --0.0031 & 0.0441 & 94.0    & & --0.0072 & 0.0438 & 93.4  & & --0.0072 & 0.0438 & 93.4 \\
\textbf{IPWRA} & --0.0033 & 0.0410 & 93.5  & & --0.0034 & 0.0411 & 93.9  & & --0.0033 & 0.0409 & 93.7 \\
\textbf{AIPW} & --0.0033 & 0.0411 & 93.6  & & --0.0033 & 0.0411 & 93.8  & & --0.0032 & 0.0409 & 93.8 \\
&       &       &       &       &       &       &       &       &       &  \\
\textbf{N=4{,}000} &       &       &       &       &       &       &       &       &       &  \\
\cline{1-1}
\textbf{IV} & 0.0028 & 0.0199 & 95.0    & & --0.0017 & 0.0198 & 95.6  & & 0.0028 & 0.0199 & 95.0 \\
\textbf{RA} & --0.0001 & 0.0191 & 95.1  & & 0.0000 & 0.0192 & 95.5  & & --0.0001 & 0.0191 & 95.1 \\
\textbf{IPW} & --0.0001 & 0.0199 & 95.6  & & --0.0042 & 0.0203 & 95.1  & & --0.0042 & 0.0203 & 95.1 \\
\textbf{IPWRA} & --0.0002 & 0.0192 & 95.7  & & --0.0002 & 0.0192 & 95.6  & & --0.0002 & 0.0192 & 95.2 \\
\textbf{AIPW} & --0.0002 & 0.0192 & 95.6  & & --0.0002 & 0.0191 & 95.6  & & --0.0002 & 0.0192 & 95.3 \\
\hline
\end{tabular}
\begin{footnotesize}
\begin{tablenotes}[flushleft]
\item \textit{Notes:} The details of the simulation design are provided in Section \ref{simulations}. Results are based on 1,000 replications. ``RMSE'' is the root mean squared error of an estimator. ``Cov.''~is the coverage rate for a nominal 95\% confidence interval. ``IV'' is the IV estimate of the coefficient on the endogenous treatment, controlling for $X$\@.  The remaining estimators are defined in the main text. To calculate the coverage rate, we use robust standard errors (IV) or standard errors that follow from the GMM framework in Section \ref{inference} (remaining estimators).
\end{tablenotes}
\end{footnotesize}
\end{threeparttable}
\end{adjustwidth}
\end{table}

Table \ref{tab:simresults2} revisits the same measures of estimator performance as Table \ref{tab:simresults1} while focusing on the binary outcome. Unlike in Table \ref{tab:simresults1}, the differences between the estimators that we consider are very minor, even in cases when one of the underlying models is misspecified. This is likely due to the fact that, as shown in Table \ref{tab1} in the Appendix, age squared is insignificant in the logit regressions of the binary outcome, IRA participation. It follows that omitting this variable should have a smaller effect on estimator performance, relative to Table \ref{tab:simresults1}.

In any case, even though the differences in estimator performance, as reported in Table \ref{tab:simresults2}, are very minor, it is also clear that the performance of RA, IPWRA, and AIPW is better overall than that of IV and IPW\@. This is again reassuring, given our general preference for IPWRA estimation.

\section{Conclusion}
\label{conclusion}

In this paper we develop a framework for doubly robust (DR) estimation of local average treatment effects, which uses quasi-likelihood methods weighted by the inverse of the instrument propensity score. These estimators are commonly referred to as inverse probability weighted regression adjustment (IPWRA)\@. We argue that our estimators have appealing small sample properties relative to competing methods, such as augmented inverse probability weighting (AIPW)\@. We discuss inference for IPWRA estimators and propose a DR version of a Hausman test previously suggested by \citetalias{DHL2014b}, which compares two estimates of the average treatment effect on the treated (ATT) in settings with one-sided noncompliance.

We discuss two empirical applications. First, we revisit \cite{Abadie2003}'s study of the effects of 401(k) retirement plans, and demonstrate that some of the conclusions are different dependent on whether one uses AIPW or IPWRA, which are the two major classes of DR estimators. While we obviously do not know which estimate is closer to the true effect of interest, we note that our preferred estimate is much more precise. Second, we reanalyze \cite{TAWBF2014}'s sample from the Oregon Health Insurance Experiment. Focusing on the effect of Medicaid on emergency room visits, we provide evidence that the local average treatment effect on the treated (LATT) is larger than the usual local average treatment effect (LATE), at least along the extensive margin. We conclude the paper with a Monte Carlo study that demonstrates the very good finite sample properties of our proposed IPWRA estimator.

\pagebreak

\section*{Appendix}

\setcounter{table}{0}
\renewcommand{\thetable}{A\arabic{table}}
\setcounter{figure}{0}
\renewcommand{\thefigure}{A\arabic{figure}}

\vspace{1cm}

\begin{table}[!h]
\begin{adjustwidth}{-1.25in}{-1.25in}
\centering
\begin{threeparttable}
\caption{Coefficient Values for the Data-Generating Process in Section \ref{simulations}\label{tab1}}
\begin{tabular}{l*{9}{c}}
\hline\hline
	& 401(k) & & 401(k) & & \multicolumn{2}{c}{Net total} & & \multicolumn{2}{c}{IRA} \\
	& eligibility & & participation & & \multicolumn{2}{c}{financial assets} & & \multicolumn{2}{c}{participation} \\
\cline{2-2}
\cline{4-4}
\cline{6-7}
\cline{9-10}
	& (1) & & (2) & & (3) & (4) & & (5) & (6) \\
\hline
Household income &   0.0000232 & &   0.0000154 & &       1.134&       0.762 & &   0.0000318&   0.0000342\\
	&     (23.00) & &      (9.18) & &     (25.67)&     (23.92) & &     (18.87)&     (20.87)\\
Age (minus 25)               &      0.0581 & &     --0.0285 & &      --106.6&      --557.4 & &      0.0420&      0.0665\\
	&      (7.25) & &     (--2.03) & &     (--0.25)&     (--2.42) & &      (2.63)&      (4.99)\\
Age (minus 25) squared            &    --0.00158 & &    0.000699 & &       41.36&       38.28 & &    0.000211&   --0.000267\\
	&     (--7.45) & &      (1.88) & &      (3.68)&      (6.32) & &      (0.52)&     (--0.82)\\
Constant            &      --1.727 & &       0.387 & &    --36377.2&    --19452.5 & &      --3.148&      --3.653\\
	&    (--24.55) & &      (3.08) & &     (--9.68)&    (--10.03) & &    (--20.02)&    (--28.04)\\
 & & & & & & & & & \\
Observations        &        9,275 & &        3,637 & &        3,637&        5,638 & &        3,637&        5,638\\
Sample & Full & & $Z=1$ & & $Z=1$ & $Z=0$ & & $Z=1$ & $Z=0$\\
Method & Logit & & Logit & & OLS & OLS & & Logit & Logit \\
\hline
\end{tabular}
\begin{footnotesize}
\begin{tablenotes}[flushleft]
\item \textit{Notes:} The table presents coefficient estimates obtained using \cite{Abadie2003}'s subsample of the Survey of Income and Program Participation (SIPP) of 1991, as previously analyzed in Table \ref{tab:table2}. The instrument, $Z$, is an indicator for 401(k) eligibility. The estimates in column (1) are used as coefficient values for equation \eqref{dgp:z}. The estimates in column (2) are used as coefficient values for equation \eqref{dgp:d1}. The estimates in columns (3) and (4) are used as coefficient values for equation \eqref{dgp:conty}. The estimates in columns (5) and (6) are used as coefficient values for equation \eqref{dgp:biny}. $t$ statistics are in parentheses.
\end{tablenotes}
\end{footnotesize}
\end{threeparttable}
\end{adjustwidth}
\end{table}

\pagebreak
\singlespacing

\setlength\bibsep{0pt}
\bibliographystyle{aer}
\bibliography{SUW_references}

\end{document}